\documentclass[journal]{IEEEtran}
\usepackage{amssymb}
\usepackage{amsmath}
\usepackage{stfloats}
\usepackage{graphicx}
\usepackage{subfigure}
\usepackage{tabularx}
\usepackage{epsfig,epsf,color,balance,cite}
\usepackage[normalem]{ulem}
\usepackage{bbding}
\usepackage{multirow}

\newtheorem{theorem}{Theorem}

\newtheorem{remark}{Remark}

\begin{document}
\title{Adaptive Spatial Modulation for Visible Light Communications with an Arbitrary Number of Transmitters}
\author{Jin-Yuan Wang,~\IEEEmembership{Member,~IEEE,}
        Hong Ge,
        Jian-Xia Zhu,
        Jun-Bo Wang, ~\IEEEmembership{Member,~IEEE,}\\
        Jianxin Dai, ~\IEEEmembership{Member,~IEEE,}
        and Min Lin, ~\IEEEmembership{Member,~IEEE}
\thanks{Jin-Yuan Wang and Min Lin are with Key Lab of Broadband Wireless Communication and Sensor Network Technology, Nanjing University of Posts and Telecommunications, Nanjing 210003, China, and also with National Mobile Communications Research Laboratory, Southeast University, Nanjing 210096, China.}% <-this % stops a space
\thanks{Hong Ge and Jianxin Dai are with School of Science, Nanjing University of Posts and Telecommunications, Nanjing 210003, China.} \thanks{Jian-Xia Zhu and Jun-Bo Wang are with National Mobile Communications Research Laboratory, Southeast University, Nanjing 210096, China.}
\thanks{Manuscript received XX XX, 2018; revised XX XX, 2018.}
}

\markboth{IEEE Access,~Vol.~XX, No.~XX, XX~2018}%
{Shell \MakeLowercase{\textit{et al.}}: Bare Demo of IEEEtran.cls for Journals}

\maketitle

\begin{abstract}
As a power and bandwidth efficient modulation scheme,
the optical spatial modulation (SM) technique has recently drawn increased attention in the field of visible light communications (VLC).
To guarantee the number of bits mapped by the transmitter's index at each timeslot is an integer, the number of transmitters (i.e., light-emitting diodes) in the SM based VLC system is often set be a power of two.
To break the limitation on the required number of transmitters and provide more design flexibility,
this paper investigates the SM based VLC with an arbitrary number of transmitters.
Initially, a channel adaptive bit mapping (CABM) scheme is proposed,
which includes three steps: bit mapping in space domain, bit mapping in signal domain, and the channel adaptive mapping.
The proposed CABM scheme allows operation with an arbitrary number of transmitters,
and is verified to be an efficient scheme through numerical results.
Based on the CABM scheme, the information-theoretical aspects of the SM based VLC are analyzed.
The theoretical expression of the mutual information is first analyzed.
However, it is very hard to evaluate system performance.
To obtain more insights, a lower bound of the mutual information is derived, which is in closed-form.
Both theoretical analysis and numerical results show that the gap between the mutual information and its lower bound is small.
Finally, to further improve the system performance, the precoding scheme is proposed for the SM based VLC.
Numerical results show that the system performance improves dramatically when using the proposed precoding scheme.
\end{abstract}

\begin{IEEEkeywords}
Adaptive spatial modulation,
Arbitrary number of transmitters,
Channel adaptive bit mapping,
Mutual information,
Precoding,
Visible light communications.
\end{IEEEkeywords}

\IEEEpeerreviewmaketitle

\section{Introduction}
\label{section1}
Multi-input multi-output (MIMO) technique can significantly increase the transmission rate and link reliability of the system without occupying additional system bandwidth, and is a promising candidate for future 5G wireless communications \cite{BIB1}.
However, the MIMO technique suffers from inter-channel interference and synchronization problems \cite{BIB1_1}.
Moreover, multiple radio frequency (RF) chains in MIMO systems result in an increase in computational complexity and cost.
Against this background, spatial modulation (SM) technique has been proposed as a solution with low complexity \cite{BIB2}.

For SM based communication system, multiple antennas are deployed at the transmitter,
but only one of them is active at each timeslot to transmit information and the others are silent.
The source data bits are divided into two parts,
one part is mapped onto the conventional constellation points in signal domain,
and the other part is used to determine the index of the active antenna in space domain \cite{BIB03}.
Therefore, the index of the antenna is utilized as an additional dimension to transfer information,
which helps to improve the data rate.
Moreover, only one RF chain is employed in SM based system,
which can effectively reduce the system complexity.

Recently, the SM technique has been extended to the optical wireless communication (OWC) field for both the outdoor and indoor
scenarios.
For outdoor OWC (also named as free-space optical (FSO) communications), the performance for the SM based system has been investigated.
In \cite{BIB3}, an analytical framework was provided for both uncoded and coded outdoor SM in FSO communication channels.
In \cite{BIB4}, the average bit error probability (ABEP) was analyzed for SM based FSO system over H-K turbulence channels.
In \cite{BIB4_1}, the constrained capacity was maximized for power-imbalanced optimal SM MIMO system.
For the indoor environment, the OWC can be applied to the visible light communications (VLC).
The optical SM in VLC was first proposed in \cite{BIB04} and then extended to many scenarios.
In \cite{BIB05}, the ABEP of SM combining pulse amplitude modulation (PAM) and space shift keying (SSK) in VLC was investigated.
In \cite{BIB06}, the mutual information of the SM based VLC system was derived.
Moreover, ref. \cite{BIB07} proposed an enhancement in SM performance by aligning the light-emitting diodes (LEDs) and the photodiodes (PDs).
In \cite{BIB08}, the performance of SM was compared with that of optical spatial multiplexing and optical repetition coding.
The effect of synchronization error on optical SM was investigated in \cite{BIB09}.
In \cite{BIB10}, an active-space, collaborative constellation-based generalized SM MIMO encoding scheme was proposed.
In \cite{BIB10_0}, a constellation optimization scheme was proposed for indoor SM based VLC.
For a complete discussion about the concept of SM and its variants,
the readers can refer to \cite{BIB11}.
It should be emphasized that the design flexibility in \cite{BIB3,BIB4,BIB4_1,BIB04,BIB05,BIB06,BIB07,BIB08,BIB09,BIB10,BIB10_0,BIB11} is limited by the fact that the number of LEDs must be a power of two to guarantee that the number of bits mapped by the transmit LED index at each timeslot is an integer.
That is, when the number of LEDs is not a power of two, the design methods in \cite{BIB3,BIB4,BIB4_1,BIB04,BIB05,BIB06,BIB07,BIB08,BIB09,BIB10,BIB10_0,BIB11} are not available.

In conventional RF wireless communications, some bit mapping schemes have been proposed when the number of transmit antennas is not a power of two.
A fractional bit encoded spatial modulation (FBE-SM) was proposed in \cite{BIBadd1}.
Unfortunately, the FBE-SM scheme suffers from the error propagation problem.
In \cite{BIBadd2}, a joint mapped spatial modulation scheme was proposed.
However, the number of transmit antennas must satisfy an equality constraint.
In addition, the bit mapping schemes with an arbitrary number of transmit antennas were also analyzed in \cite{BIB14} and \cite{BIB15}.
However, the channel state information at the transmitter (CSIT) was not considered in \cite{BIBadd1,BIBadd2,BIB14,BIB15}.
By considering the CSIT, link adaptive mapper designs were proposed in \cite{BIBadd3} and \cite{BIBadd4}.
In \cite{BIBadd3}, the CSIT was employed to design the SSK modulated system.
Note that the SSK modulation is a special case of SM, where only antenna index is used to convey information.
Therefore, the bit mapping in signal domain was not considered in \cite{BIBadd3}.
In \cite{BIBadd4}, an adaptive brute forth mapper (BFM) was proposed for SM with lightweight feedback overhead.
However, the number of transmit antennas must be a power of two.

Motivated by the above literature,
this paper investigates the channel adaptive SM for indoor VLC to break the limitation on the required number of LEDs. That is, the CSIT is utilized for bit mapping, and the number of LEDs can be an arbitrary positive number larger than one.
The main contributions of this paper are listed as follows:
\begin{itemize}
  \item With an arbitrary number of LEDs, a channel adaptive bit mapping (CABM) scheme for SM based VLC is proposed.
        The proposed CABM scheme is a modification of the schemes in \cite{BIB14} and \cite{BIB15} by using the CSIT, which includes three steps. The first two steps are the bit mappings in space domain and signal domain, respectively. At last, by using CSIT, a constellation optimization is formulated to obtain better modulation combinations. In the CABM scheme, when the number of LEDs is not a power of two, space domain symbols are mapped with different numbers of bits. To keep the same number of transmit bits at each timeslot, the constellations in signal domain are also mapped with different modulation orders.
  \item After performing the proposed CABM scheme, the mutual information is analyzed for SM based VLC with an arbitrary number of LEDs.
        In the system, the finite alphabet and input-dependent noise are employed. According to Shannon's information theory, the theoretical expression of the mutual information is derived, but is with an integral expression and cannot be easily used in practice. To reduce the complexity, a closed-form expression of the lower bound on the mutual information is obtained.
  \item To further improve the system performance, a precoding scheme is proposed for the SM based VLC.
        The purpose of the precoding is to maximize the minimum distance between any two constellation points in the received signal space.
        The optimization problem is non-convex and non-differentiable, which is very hard to obtain the optimal solution. Alternatively, an approximation is employed for the minimum function, and the optimization problem is then solved by using the interior point algorithm.
\end{itemize}

The rest of the paper is organized as follows.
In Section \ref{section2}, the system model is described.
In Section \ref{section3}, the CABM scheme is proposed.
Section \ref{section4} analyzes the mutual information and its lower bound.
A precoding scheme is proposed in Section \ref{section5}.
Numerical results are shown in Section \ref{section6}.
Section \ref{section7} draws conclusions of this paper.

\emph{Notations}:
Throughout this paper, italicized symbols denote scalar values,
and bold symbols denote vectors/matrices/sets.
We use ${\mathbb{N}^ + }$ for positive integer,
$\left\langle  \cdot  \right\rangle $ for the inner product,
$\mathbb{E(\cdot)}$ for the expect operator,
${\left\|  \cdot  \right\|_{\rm{F}}}$ for the Frobenius norm,
${\rm Pr }(\cdot)$ for the probability of an event.
We use ${\cal N}\left( {0,{\sigma ^2}} \right)$ for a Gaussian distribution with zero mean and variance ${\sigma ^2}$,
${\# _\mathbf{\Delta} }\left( a \right)$ for the number of $a$ in set $\mathbf{\Delta}$,
and $p( \cdot )$ for the probability density function (PDF) of a random variable.
We use ${\log _2}( \cdot )$ for the logarithm with base 2,
and $\ln ( \cdot )$ for the natural logarithm.
We use ${\cal I}(\cdot;\cdot)$ for the mutual information,
${\cal H}(\cdot)$ for the entropy, and $\mathcal{Q}(\cdot)$ for the Gaussian-Q function.

\section{System Model}
\label{section2}
Consider an SM based indoor VLC system with $M$ transmitters (i.e., LEDs) and one receiver (i.e., PD),
as illustrated in Fig. \ref{fig1}.
In the system, $M$ is an arbitrary positive number larger than one.
With finite alphabet ($K$ bits), the source data is divided into two parts by using optical SM.
One part is mapped onto the indexes of the LEDs in space domain,
and the other part is mapped onto the constellation points in signal domain.

\begin{figure}
\centering
\includegraphics[width=9cm]{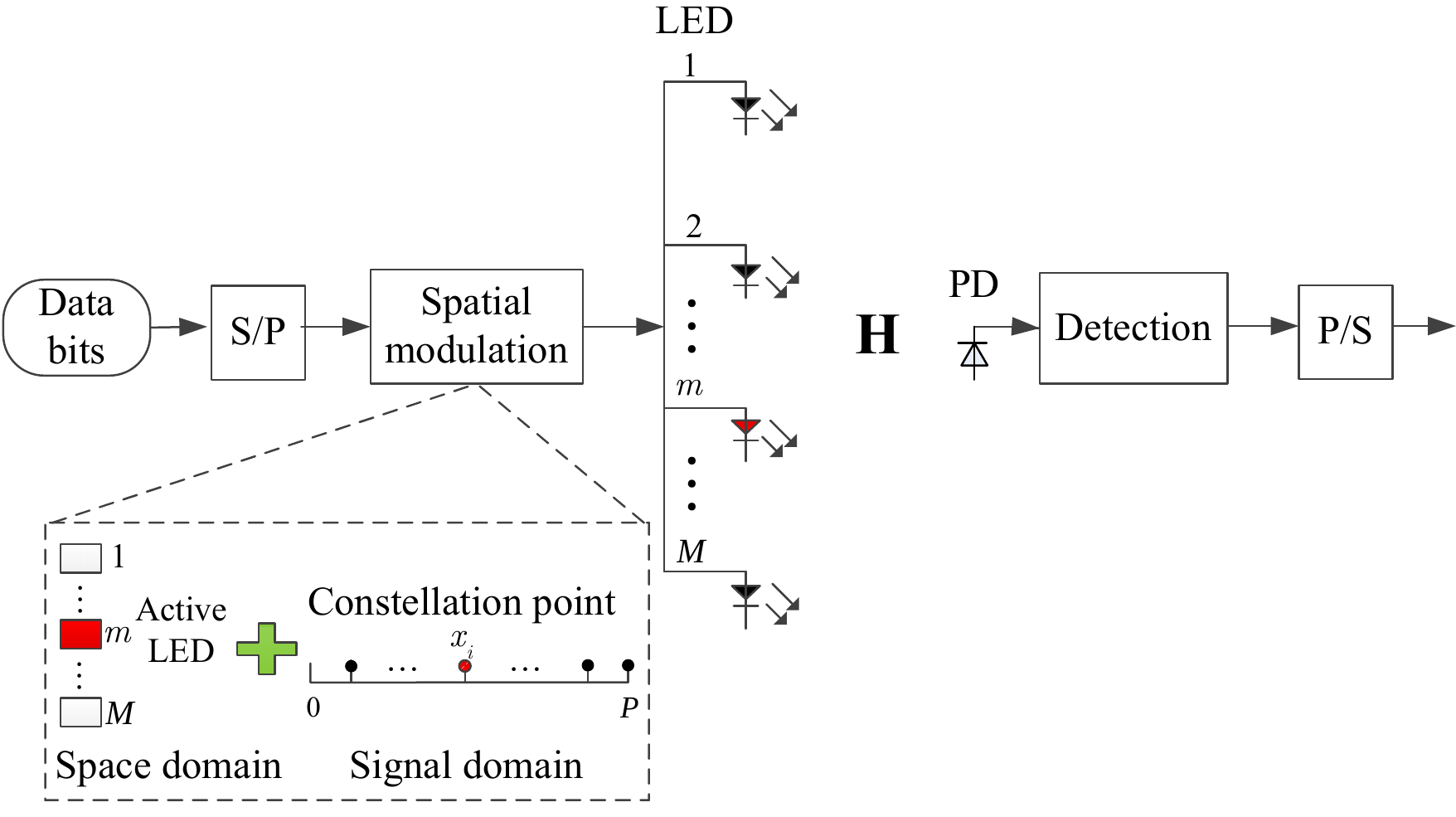}
\caption{An SM based VLC system.}
\label{fig1}
\end{figure}

By employing SM, only one LED is active at each timeslot.
We assume that the $m$-th LED is selected at the current timeslot to transmit the $i$-th symbol.
According to \cite{BIB12}, the received signal at the PD is given by
\begin{equation}
y = {h_m}{x_i} + \sqrt {{h_m}{x_i}} {z_1} + {z_0},
\label{eq1}
\end{equation}
where ${x_i}$ denotes the $i$-th transmitted optical intensity signal in signal domain.
The average power of $x_i$ is denoted as $P_t$.
$z_0$ and $z_1$ are independent of each other, where ${z_0} \sim {\cal N}\left( {0,{\sigma ^2}} \right)$ is the input-independent Gaussian noise,
and ${z_1} \sim {\cal N}\left( {0,{\varsigma ^2}{\sigma ^2}} \right)$ is the input-dependent Gaussian noise.
${\varsigma ^2} > 0$ denotes the ratio of the input-dependent noise variance to the input-independent noise variance.
Moreover, ${h_m} \in {\bf{H}} = \{ {h_1},{h_2}, \cdots ,{h_M}\} $ represents the real-valued channel gain between the $m$-th LED and the PD, which is given by \cite{BIB12_1,BIB13}
\begin{equation}
{h_m} \!\!=\!\! \left\{\!\!\!\!\! {\begin{array}{*{20}{c}}
{\frac{{(l + 1)E}}{{2\pi d_m^2}}{{\cos }^l}({\phi _m})\cos ({\varphi _m}),}& {\rm if}\; {0 \le {\varphi _m} \le {\Psi _{\rm{c}}}}\\
{0,}&{\rm if}\; {{\varphi _m} > {\Psi _{\rm{c}}}}
\end{array}} \right.,
\label{eq2}
\end{equation}
where $E$ is the physical area of the PD.
$l =  - \ln 2/\ln \left( {\cos {\Phi _{1/2}}} \right)$ is the Lambertian emission order, and ${\Phi _{1/2}}$ is the semi-angle at half-power of the LED.
${d_m}$ is the distance between the $m$-th LED and the PD.
${\phi _m}$ and ${\varphi _m}$ are the emission angle and the incidence angle from the $m$-th LED to the PD, respectively.
${\Psi _c}$ denotes the field of view (FOV) of the PD.
By using SM, only one element in ${\bf{H}}$ is non-zero at each timeslot.

\section{Channel Adaptive Bit Mapping}
\label{section3}
In this paper, the number of LEDs $M$ is arbitrary.
When $M = {2^p},\;p \in {\mathbb{N}^ + }$,
$p$ bits can be mapped onto the LEDs' indexes at each timeslot using the traditional bit mapping scheme \cite{BIB13_1}. However, when $M \ne {2^p},\;p \in {\mathbb{N}^ + }$, the traditional bit mapping scheme is not available.
In this section, a CABM scheme is provided for bitting mapping.

The principle of the CABM scheme is presented in Fig. \ref{fig2}.
The CABM scheme includes three steps.
The first two steps are the bit mapping in space domain and the bit mapping in signal domain, respectively.
Because the number of the LEDs is not a power of two,
the LEDs' indexes in space domain are mapped with different numbers of bits.
To keep the same total bit length at each timeslot (i.e., $K$ bits),
different modulation modes are employed in signal domain.
Assuming that the CSIT is known,
the last step is to adaptively select the best modulation order combination in signal domain to get the optimal bit mapping.
To facilitate the understanding, the detailed descriptions about the CABM scheme are given in the following three subsections.

\begin{figure}
\centering
\includegraphics[width=8cm]{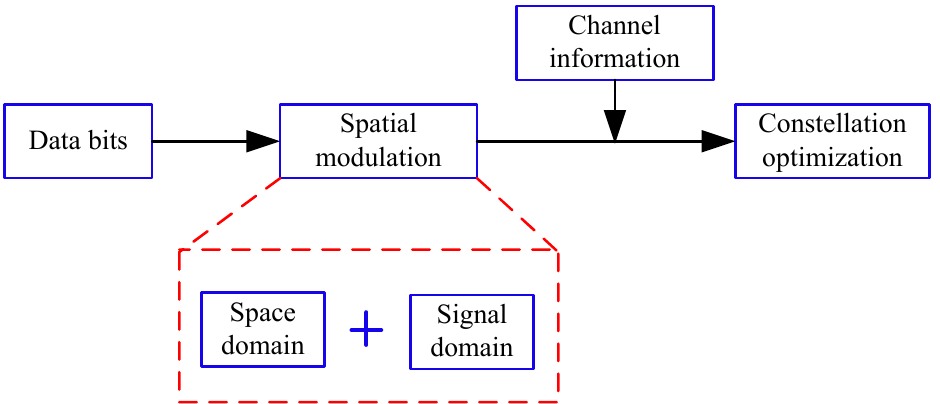}
\caption{The channel adaptive bit mapping scheme.}
\label{fig2}
\end{figure}

\subsection{Bit Mapping in Space Domain}
\label{section3_1}
When the number of LEDs $M$ is not a power of two,
an integer $p$ can be found such that ${2^p} < M < {2^{p + 1}}$.
After that, prescribe some information bits transmitted by LEDs with a length of $p$ bits and the others with a length of $p + 1$ bits.
Referring to \cite{BIB14} and \cite{BIB15}, the bit mapping in space domain can be implemented by using the following procedures.
\begin{enumerate}
  \item Choose $p \in {\mathbb{N}^ + }$ such that ${2^p} < M < {2^{p + 1}}$.
  \item Define $\mathbf{\Gamma} = \{ 1,2,...,M\}$ be the set of the LEDs' indexes,
        and select the former ${2^p}$ elements from $\mathbf{\Gamma}$ into set $\mathbf{\Omega}$.
        Therefore, $\mathbf{\Omega} =\{1,2,\cdots, 2^p\}$.
  \item Divide set $\mathbf{\Omega}$ into two sets $\mathbf{\Psi}$ and $\mathbf{\Xi}$,
        where $\mathbf{\Xi}  = \left\{ {1,2, \cdots ,{2^{p + 1}} - M} \right\}$ contains the former ${2^{p + 1}} - M$ elements
        and $\mathbf{\Psi}  = \left\{ {{2^{p + 1}} - M + 1, \cdots ,{2^p}} \right\}$ contains the remainder $M - {2^p}$ elements.
      Then, $p$ bits are used to map these ${2^p}$ elements.
  \item Define set $\mathbf{\Phi}  = \{{2^p} + 1,{2^p} + 2,...,M\}$ containing the unsettled $M - {2^p}$ LEDs. The bit mapping in set $\mathbf{\Phi}$ is initially the same as that in set $\mathbf{\Psi}$.
  \item ``0" is appended at the end of each bit mapping for set $\mathbf{\Psi}$, and ``1" is appended at the end of each bit mapping for set $\mathbf{\Phi}$. So far,  $p + 1$ bits are mapped in sets $\mathbf{\Psi}  \cup \mathbf{\Phi}$ and $p$ bits are mapped in set $\mathbf{\Xi}$.
\end{enumerate}

\subsection{Bit Mapping in Singal Domain}
\label{section3_2}
After performing the bit mapping in space domain, the bit lengths mapped on each LED are not the same.
To ensure the total number of bits transmitted at each timeslot remains the same (i.e., $K$ bits),
different modulation modes should be employed in signal domain.
Without loss of generality, the PAM is employed for the VLC system.
Therefore, the signals with higher-order (${2^q}$-ary) PAM are transmitted from the LEDs pertaining to $p$ bits
and those with lower-order (${2^{q - 1}}$-ary) PAM are transmitted from the LEDs pertaining to $p+1$ bits.
Briefly, the bit mapping in signal domain is performed by using the following procedures.
\begin{enumerate}
  \item Choose $q=K-p$.
  \item When the LEDs with the index in sets $\mathbf{\Psi}$ and $\mathbf{\Phi}$ are activated, ${2^{q - 1}}$-ary PAM is employed to transmit information.
  \item When the LEDs with the index in set $\mathbf{\Xi}$ are activated, ${2^q}$-ary PAM is employed to transmit information.
\end{enumerate}

\subsection{Channel Adaptive Mapping}
\label{section3_3}
In \cite{BIB14} and \cite{BIB15}, the CSIT is not considered for bit mapping.
In this case, LED with bad channel state may employ high-order signal modulation, which will degrade system performance.
To improve system performance, the channel adaptive mapping is employed.

In VLC, typical indoor illumination environments offer very high signal-to-noise ratio (SNR) \cite{BIB16}.
For the channel in (\ref{eq1}), when the $m$-th LED is activated to transmit the symbol $x_i$,
the conditional pairwise error probability (PEP) at high SNR is given by \cite{BIB17}
\begin{eqnarray}
&&\!\!\!\!\!\!\!\!\!\!\!\!\!\!\!\! PEP\left( {{x_i} \to {x_j}|{h_m}} \right) \nonumber\\
%&&\!\!\!\!\!\!\!\!\!\!\!\!\!\!\!\!={\rm Pr} \left\{ {\left\| {y - {h_m}{x_j}} \right\|_{\rm{F}}^2 < \left\| {y - {h_m}{x_i}} \right\|_{\rm{F}}^2} \right\}\nonumber \\
&&\!\!\!\!\!\!\!\!\!\!\!\!\!\!\!\!= {\rm Pr}\left\{ {\left( {\sqrt {{h_m}{x_i}} {z_1} + {z_0}} \right) \cdot ({h_m}{x_i} - {h_m}{x_j}) > d_{ij}^2/2} \right\},
 \label{eq3}
\end{eqnarray}
where ${d_{ij}} = {\left\| {{h_m}{x_i} - {h_m}{x_j}} \right\|_{\rm{F}}}$. Note that $z \buildrel \Delta \over = \sqrt {{h_m}{x_i}} {z_1} + {z_0} \sim \mathcal{N}\left( {0,\left( {{h_m}{x_i}{\varsigma ^2} + 1} \right){\sigma ^2}} \right)$ for a given ${h_m}$, and thus $z \cdot \left( {{h_m}{x_i} - {h_m}{x_j}} \right) \sim \mathcal{N}\left( {0,({h_m}{x_i}{\varsigma ^2} + 1){\sigma ^2}d_{i,j}^2} \right)$. Therefore, the conditional PEP is given by
\begin{equation}
PEP\left( {{x_i} \to {x_j}|{h_m}} \right) = {\cal Q}\left( {\frac{{{d_{ij}}}}{{2\sqrt {1 + {h_m}{x_i}{\varsigma ^2}} \sigma }}} \right).
\label{eq4}
\end{equation}

\begin{remark}
It can be observed from (\ref{eq4}) that the conditional PEP $PEP\left( {{x_i} \to {x_j}|{h_m}} \right)$ is a monotonically decreasing function with respect to $h_m$.
Moreover, it can be seen from (\ref{eq2}) that $h_m$ is a monotonicity decreasing function with respect to the angle of emission ${\phi _m}$.
This indicates that the system error performance degrades with the increase of ${\phi _m}$.
\label{rem0}
\end{remark}

To verify \emph{Remark \ref{rem0}}, the BER performance in the receive plane of the PD is provided in Fig. \ref{addfig1} when $P_t=25$ dBm, $M=1$, $K=5$ and the coordinates of the LED and the PD are (2.5m, 2m, 3m) and ($x$m, $y$m, 0.8m). It can be found that when the PD is located at (2.5m, 2m, 0.8m), the angle of emission ${\phi _m}$ is zero, the best BER performance achieves.
When the PD moves far away from (2.5m, 2m, 0.8m), the angel of emission ${\phi _m}$ increase,
and the value of the BER increases, which verifies the conclusion in \emph{Remark \ref{rem0}}.

\begin{figure}
\centering
\includegraphics[width=8cm]{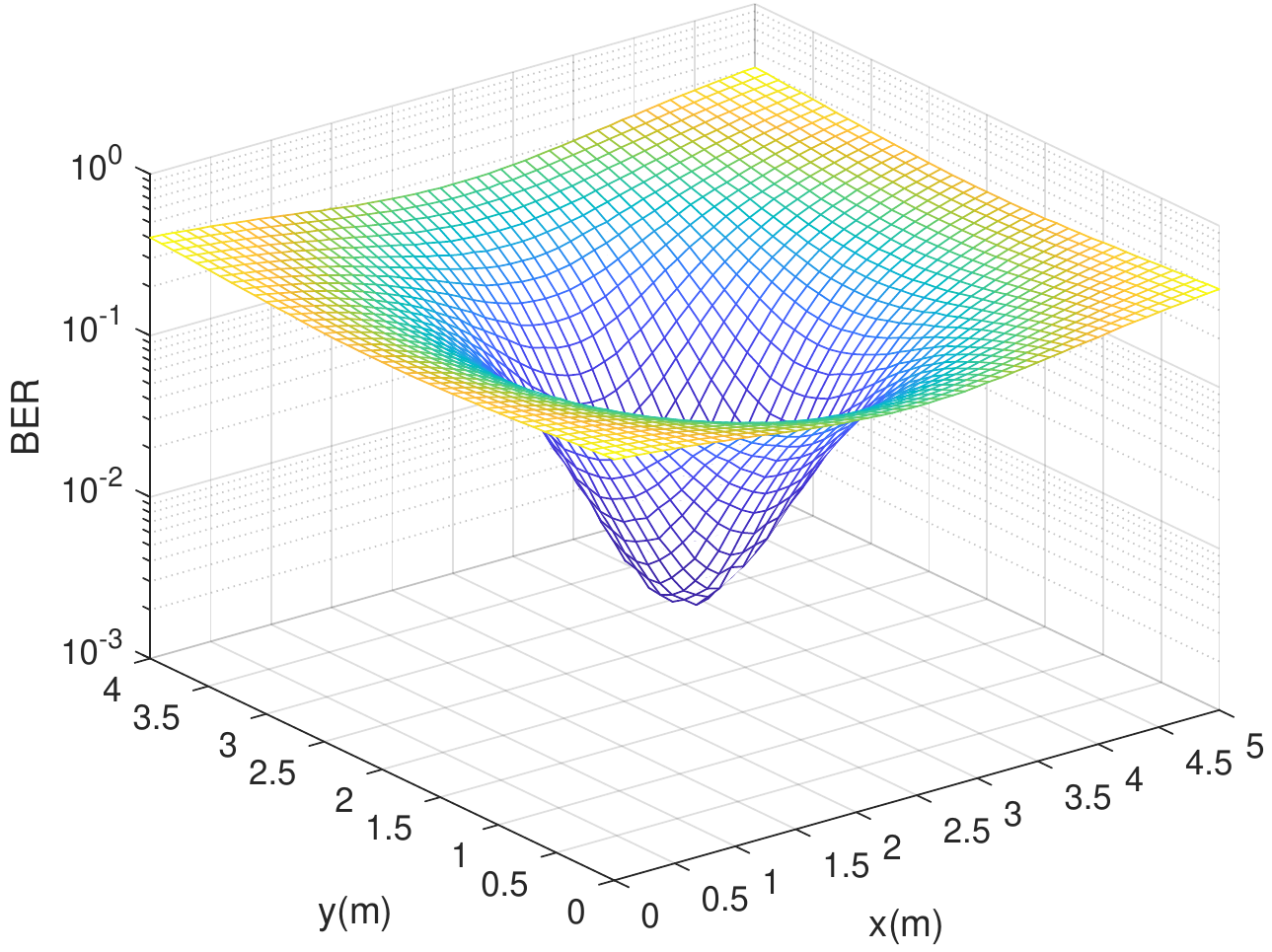}
\caption{BER performance at the receiver plane.}
\label{addfig1}
\end{figure}

Moreover, it can be observed from (\ref{eq4}) that the system error performance is dominated by the term ${d'_{i,j}} = {{{d_{ij}}} \mathord{\left/
 {\vphantom {{{d_{ij}}} {\sqrt {1 + {h_m}{x_i}{\varsigma ^2}} }}} \right.
 \kern-\nulldelimiterspace} {\sqrt {1 + {h_m}{x_i}{\varsigma ^2}} }}$ among constellation points.
The minimum value of ${d'_{i,j}}$ is given by
\begin{equation}
{d'_{\min }}({\bf{H}}) = \mathop {\min }\limits_{\scriptstyle{x_i},{x_j} \in \mathbf{\Lambda} \hfill\atop
\scriptstyle{x_i} \ne {x_j}\hfill} \frac{{{{\left\| {{\bf{H}}({x_i} - {x_j})} \right\|}_{\rm{F}}}}}{{\sqrt {1 + {h_m}{x_i}{\varsigma ^2}} }},
\label{eq5}
\end{equation}
where ${\bf{\Lambda }}$ denotes the set of all possible transmit signals.
To improve the system performance, the minimum term ${d'_{\min }}({\bf{H}})$ should be maximized.
Consequently, the modulation order optimization problem can be formulated as
\begin{equation}
  \begin{split}
    & [{{\tilde N}_1},{{\tilde N}_2}, \ldots ,{{\tilde N}_M}] = \arg \mathop { {\rm{  max}}}\limits_{{N_m},m \in \bf{\Gamma}} {\rm{  }}{d'_{\min }}\left( {\bf{H}} \right) \\
    \text{s.t.}\qquad & \left\{ {\begin{array}{*{20}{c}}
{{N_m} \in \mathbf{\Delta}  = \left\{ {{2^{q - 1}},{2^q}} \right\},\;\;m \in \bf{\Gamma} }\\
{{\# _\mathbf{\Delta} }\left( {{2^{q - 1}}} \right) = 2\left( {M - {2^p}} \right)}\\
{{\# _\mathbf{\Delta} }\left( {{2^q}} \right) = {\rm{ }}{2^{p + 1}} - M}
\end{array}} \right.,
  \end{split}
\label{eq6}
\end{equation}
where ${\tilde N_1},{\tilde N_2}, \cdots ,{\tilde N_M}$ are the optimal modulation orders in signal domain on all LEDs.
${\# _\mathbf{\Delta} }\left( {{2^{q - 1}}} \right)$ and ${\# _\mathbf{\Delta} }\left( {{2^q}} \right)$ denotes the total numbers of ${2^{q - 1}}$ and ${2^q}$ in set $\mathbf{\Delta}$, respectively.

In the optimization problem (\ref{eq6}), the corresponding modulation orders are the optimal solution.
However, the problem (\ref{eq6}) is an integer optimization problem,
which is non-convex and very hard to obtain the optimal solution.
From the viewpoint of implementation, computationally efficient algorithms are
more preferred. Here, Algorithm 1 is proposed to solve problem (\ref{eq6}), which is shown in Fig. \ref{fig3}.

\begin{figure}[!h]
\hrulefill\\
\small
\textbf{Algorithm 1} (Modulation order optimization algorithm) \vspace*{-5pt}\\
\vspace*{0pt}\hrulefill \\
\textbf{Step 1):} Given the positions of LEDs and the PD, calculate the channel gains by using (\ref{eq2}).\\
\textbf{Step 2):} Find all possible modulation order combinations ${\bf{D}} = \left\{ {{\mathbf{\Pi} _1},{\mathbf{\Pi} _2},...,{\mathbf{\Pi} _L}} \right\}$, where $L$ is the total number of possible combinations, ${\mathbf{\Pi} _i} = \left[ {N_1^i,N_2^i, \cdots ,N_M^i} \right]$ denotes the $i$-th modulation order combination, and   $N_m^i$ denotes the modulation order of the $m$-th LED in the $i$-th modulation order combination.\\
\textbf{Step 3):} Compute ${d'_{\min }}({\bf{H}})$ by using (\ref{eq5}) for all combinations in ${\bf{D}}$.\\
\textbf{Step 4):} Select the combination with the maximum ${d'_{\min }}({\bf{H}})$ as the output ${\mathbf{\Pi} _{{\rm{optimal}}}}$. \\
\vspace*{-5pt} \hrulefill
\caption{Modulation order optimization algorithm.}
\label{fig3}
\end{figure}

By using Algorithm 1,
the best modulation order combination ${\mathbf{\Pi }_{{\rm{optimal}}}}$ is obtained, and the corresponding set of the LEDs' indexes is $\mathbf{\Gamma}=\{1,2,\ldots,M\}$.
Sorting the elements in ${\mathbf{\Pi} _{{\rm{optimal}}}}$ by using decreasing order,
and the corresponding new set of the LEDs's indexes becomes $\mathbf{\Gamma'}$.
Then, replace $\mathbf{\Gamma}$ with $\mathbf{\Gamma'}$ in Step 2) of Section \ref{section3_1},
and perform the bit mapping in the space and signal domains again, the performance of the bit mapping can be improved.

\subsection{Case Study}
To facilitate the understanding,
an example is provided in Fig. \ref{fig4} to compare the proposed CABM scheme with previous schemes in \cite{BIB14} and \cite{BIB15}.
In Fig. \ref{fig4}, the number of LEDs is set to be $M = 6$.
Assume that the indexes of the LEDs are 1, 2, 3, 4, 5 and 6,
and the corresponding channel gains are 0.08, 0.15, 0.13, 0.25, 0.01, 0.22, respectively.
For Fig. \ref{fig4}(a), the bit mappings in space and signal domains are performed according to Sections \ref{section3_1} and \ref{section3_2}.
In the figure, ``**" represents the 2 bits mapped by 4-ary PAM modulation, which can be 00, 01, 10 or 11,
and ``$\times$" stands for the single bit mapped by 2-ary PAM, which can be 0 or 1.
However, if the CSIT is considered as in Fig. \ref{fig4}(b),
set $\mathbf{\Gamma}$ is first changed to $\mathbf{\Gamma'}$, and then
Sections \ref{section3_1} and \ref{section3_2} are performed.
Because LED 4 and LED 6 have good channel states, 4-PAM is employed in signal domain.
The other LEDs with not so good channel states are using 2-PAM.
To show the improvement of system performance, some results will be shown in Section \ref{section6_1}.

\begin{figure}
\centering
\includegraphics[width=9cm]{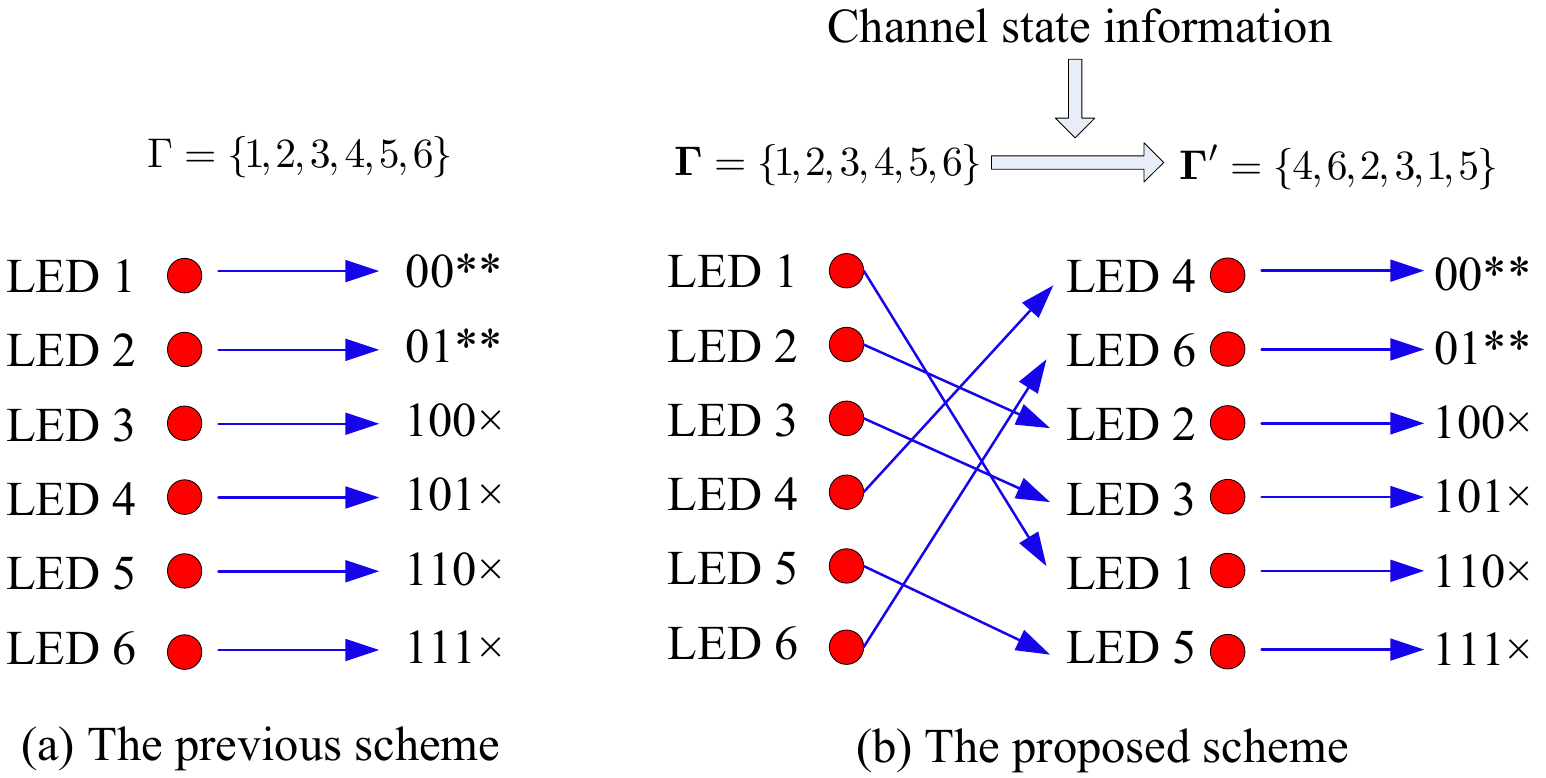}
\caption{Comparison between the previous scheme and the proposed CABM scheme when $M=6$.}
\label{fig4}
\end{figure}

\section{Mutual Information Analysis}
\label{section4}
In the above section, the CABM scheme is proposed.
Based on the CABM scheme,
the theoretical expression of the mutual information and its lower bound will be derived.
In addition, some remarks and insights are also provided.

According to the CABM scheme and the system model, the system diagram can be illustrated as Fig. \ref{fig5}.
In Fig. \ref{fig5}, the total number of the source data bits is $K=p+q$ bits.
The source data is divided into two groups by using SM.
One group is with $p$ (or $p+1$) bits to selective the LED (i.e., select the channels) in space domain,
while the other one is with $q$ (or $q-1$) bits to select the constellation point in signal domain.
Consequently, the total number of bits transmitted at any timeslot is a constant.

\begin{figure}
\centering
\includegraphics[width=8.5cm]{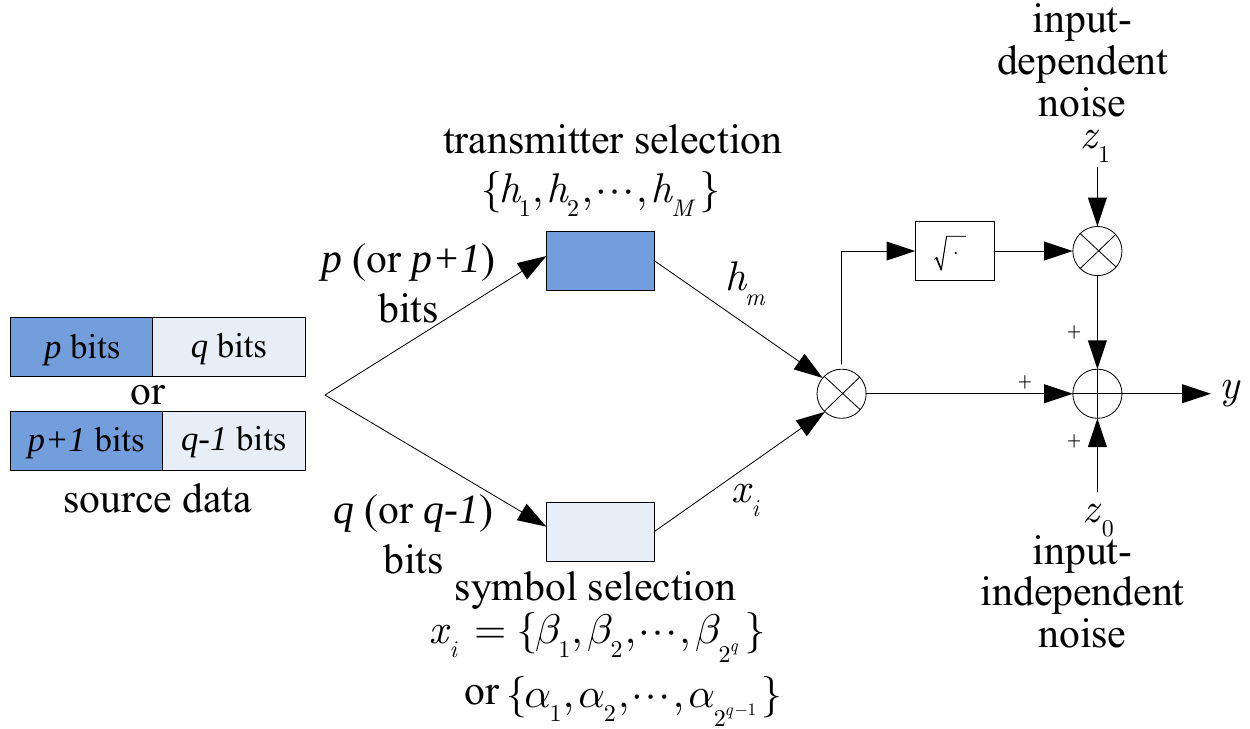}
\caption{Principle of bit mapping for SM based VLC.}
\label{fig5}
\end{figure}

Assume that the $m$-th LED is selected to transmit signal, and then $p\left( {h = {h_m}} \right)$ can be expressed as
\begin{equation}
p\left( {h = {h_m}} \right) = \left\{ \begin{array}{l}
 \frac{1}{{{2^{p + 1}}}},\;{\rm{if}}\;m \in \mathbf{\Psi}  \cup \mathbf{\Phi} \\
 \frac{1}{{{2^p}}},\;\;\;{\rm{if}}\;m \in \mathbf{\Xi}
\end{array} \right.,
\label{eq7}
\end{equation}
where the total number of elements in sets $\mathbf{\Psi}  \cup \mathbf{\Phi} $ is $2(M - {2^p})$, and the number of elements in set $\mathbf{\Xi} $ is ${2^{p + 1}} - M$.

When $m \in \mathbf{\Psi}  \cup \mathbf{\Phi} $, the space domain includes $p + 1$ bits,
and the signal domain must include $q - 1$ bits to guarantee the total number of transmit bits invariant at each timeslot.
In this case, the symbol set in signal domain is denoted as ${\bf{A}} = \left\{ {{\alpha _1},{\alpha _2}, \cdots ,{\alpha _{{2^{q - 1}}}}} \right\}$. Similarly, when $m \in \mathbf{\Xi}$, the space domain uses $p$ bits, and the signal domain must use $q$ bits. On this occasion, the symbol set in signal domain is denoted as ${\bf{B}} = \left\{ {{\beta _1},{\beta _2}, \cdots ,{\beta _{{2^q}}}} \right\}$.
Therefore, we have
\begin{equation}
p(x = {\alpha _i}|h = {h_m}) = \left\{ \begin{array}{l}
\frac{1}{{{2^{q - 1}}}},\;\;\;{\rm{if}}\;m \in \mathbf{\Psi}  \cup \mathbf{\Phi} \\
0,\;\;\;\;\;\;\;{\rm{if}}\;m \in \mathbf{\Xi}
\end{array} \right.,
\label{eq8}
\end{equation}
and
\begin{equation}
p(x = {\beta _j}|h = {h_m}) = \left\{ \begin{array}{l}
\frac{1}{{{2^q}}},\;\;{\rm{if}}\;m \in \mathbf{\Xi} \\
0,\;\;\;\;{\rm{if}}\;m \in \mathbf{\Psi}  \cup \mathbf{\Phi}
\end{array} \right..
\label{eq9}
\end{equation}
Moreover, $p(x = {\alpha _i})$ and $p(x = {\beta _j})$ can be written, respectively, as
\begin{eqnarray}
p(x = {\alpha _i}) %\!\!\!\! &=& \!\!\!\!  \sum\limits_{m \in \mathbf{\Psi}  \cup \mathbf{\Phi} } {p(h = {h_m})p(x = {\alpha _i}|h = {h_m})} \nonumber \\
 =  \frac{{M - {2^p}}}{{{2^{p + q - 1}}}},
 \label{eq10}
\end{eqnarray}
and
\begin{eqnarray}
p(x = {\beta _j})  %\!\!\!\! &=& \!\!\!\!  \sum\limits_{m \in \mathbf{\Xi} } {p(h = {h_m})p(x = {\beta _i}|h = {h_m})} \nonumber \\
 =  \frac{{{2^{p + 1}} - M}}{{{2^{p + q}}}}.
\label{eq11}
\end{eqnarray}
Then, a unified expression for (\ref{eq10}) and (\ref{eq11}) can be written as
\begin{eqnarray}
p(x = {x_i}) = \left\{ \begin{array}{l}
\displaystyle {\frac{{M - {2^p}}}{{{2^{p + q - 1}}}},\;\;\;\;{x_i} \in {\mathbf{A}}}\\
\displaystyle {\frac{{{2^{p + 1}} - M}}{{{2^{p + q}}}},\;{x_i} \in {\mathbf{B}}}
\end{array} \right..
\label{eq12}
\end{eqnarray}
According to (\ref{eq1}), if $m \in \mathbf{\Psi}  \cup \mathbf{\Phi},\;{x_i} \in {\bf{A}}$ (or $m \in \mathbf{\Xi} ,\;{x_i} \in {\bf{B}}$), the following conditional PDF is derived by
\begin{eqnarray}
p\left( {\left. y \right|h = {h_m},x = {x_i}} \right) = \frac{{\exp \left[ { - \frac{{{{(y - {h_m}{x_i})}^2}}}{{2(1 + {h_m}{x_i}{\varsigma ^2}){\sigma ^2}}}} \right]}}{{\sqrt {2\pi (1 + {h_m}{x_i}{\varsigma ^2}){\sigma ^2}} }}.
\label{eq13}
\end{eqnarray}
Furthermore, we have
\begin{eqnarray}
p\left( {y\left| {x \!=\! {x_i}} \right.} \right) \!\!=\!\!\! \left\{ \begin{array}{l}\!\!\!\!
\frac{1}{{{2^{p + 1}}}}\!\!\! \sum\limits_{m \in \mathbf{\Psi}  \cup \mathbf{\Phi} }\!\!\! {\frac{{\exp \!\left[\!\! { - \frac{{{{(y \!-\! {h_m}{x_i})}^2}}}{{2(1 \!+\! {h_m}{x_i}{\varsigma ^2}){\sigma ^2}}}}\!\! \right]}}{{\sqrt {2\pi (1 + {h_m}{x_i}{\varsigma ^2}){\sigma ^2}} }}} ,{\rm{if}}\;{x_i} \!\in\! {\bf{A}}\\
\!\!\!\! \frac{1}{{{2^p}}}\!\! \sum\limits_{m \in \mathbf{\Xi} }\!\! {\frac{{\exp \left[\!\! { - \frac{{{{(y \!-\! {h_m}{x_i})}^2}}}{{2(1 \!+\! {h_m}{x_i}{\varsigma ^2}){\sigma ^2}}}}\!\! \right]}}{{\sqrt {2\pi (1 + {h_m}{x_i}{\varsigma ^2}){\sigma ^2}} }},\;\;\;\;{\rm{if}}\;{x_i} \!\in\! {\bf{B}}}
\end{array} \right.\!\!.
\label{eq14}
\end{eqnarray}
According to (\ref{eq12}) and (\ref{eq14}), the output PDF is given by
\begin{eqnarray}
p(y)\!\!\!\! &=&\!\!\!\! \frac{{M \!-\! {2^p}}}{{{2^{2p + q}}}}\!\!\sum\limits_{{x_i} \in {\bf{A}}} {\sum\limits_{m \in \mathbf{\Psi}  \cup \mathbf{\Phi} } {\frac{{\exp\! \left[{ - \frac{{{{\left( {y - {h_m}{x_i}} \right)}^2}}}{{2\left( {1 + {h_m}{x_i}{\varsigma ^2}} \right){\sigma ^2}}}} \!\right]}}{{\sqrt {2\pi ( {1 \!+\! {h_m}{x_i}{\varsigma ^2}} ){\sigma ^2}} }}} } \nonumber\\
 &+&\!\!\!\!  \frac{{{2^{p + 1}} \!-\! M}}{{{2^{2p + q}}}}\!\!\sum\limits_{{x_i} \in {\bf{B}}} {\sum\limits_{m \in \mathbf{\Xi} } {\frac{{\exp \!\left[ { - \frac{{{{\left( {y - {h_m}{x_i}} \right)}^2}}}{{2\left( {1 + {h_m}{x_i}{\varsigma ^2}} \right){\sigma ^2}}}} \right]}}{{\sqrt {2\pi ( {1 \!+\! {h_m}{x_i}{\varsigma ^2}} ){\sigma ^2}} }}} }.
 \label{eq15}
\end{eqnarray}

\subsection{Mutual Information}
\label{section4_1}
From \cite{BIB18}, the mutual information between the input and the output can be expressed as
\begin{equation}
{\cal I}(x,h;y) = {\cal I}\left( {h;y|x} \right) + {\cal I}(x;y),
\label{eq16}
\end{equation}
where ${\cal I}\left( {h;y|x} \right)$ denotes the conditional mutual information between $h$ and $y$ when given $x$,
${\cal I}(x;y)$ is the mutual information between $x$ and $y$.

According to the channel model in (\ref{eq1}), and using the probability theory,
the mutual information is analyzed in the following theorem.

\begin{theorem}
For the SM based VLC, the theoretical expression of the mutual information is given by (\ref{eq23}) as shown at the top of the next page,
\begin{table*}\normalsize
\begin{eqnarray}
{\cal I}\left( {x,h;y} \right) \!\!\!\! &=&\!\!\!\! \frac{{{{\left( {M \!-\! {2^p}} \right)}^2} \!+\! {{\left( {{2^{p + 1}} \!-\! M} \right)}^2}}}{{{2^{2p}}}}\left( {p \!+\! 1} \right) \!+\! \frac{{M \!-\! {2^p}}}{{{2^p}}}{\log _2}\frac{{{2^{p + q - 1}}}}{{M \!-\! {2^p}}} \!+\! \frac{{{2^{p + 1}} \!-\! M}}{{{2^p}}}{\log _2}\frac{{{2^{p + q}}}}{{{2^{p + 1}} \!-\! M}}\nonumber\\
&&\!\!\!\!\!\!\!\!\! + \frac{{M \!-\! {2^p}}}{{{2^{2p + q}}}}\!\!\!\!\sum\limits_{m \in \mathbf{\Psi}  \cup \mathbf{\Phi} }\! {\sum\limits_{{x_i} \in {\mathbf{A}}} {{\mathbb{E}_z}\!\!\left[\! {{{\log }_2}\! {\frac{{\frac{{\exp \left[ { - \frac{{{z^2}}}{{2(1 + {h_m}{x_i}{\varsigma ^2}){\sigma ^2}}}} \right]}}{{\sqrt {1 + {h_m}{x_i}{\varsigma ^2}} }}}}{{\sum\limits_{{x_{{i_2}}} \in {\bf{A}}} \!{\sum\limits_{{m_2} \in \mathbf{\Psi}  \cup \mathbf{\Phi} }\!\!\! {\frac{{\exp \!\left[\! { - \frac{{{{\left( {z + d_{m,i}^{{m_2},{i_2}}} \right)}^2}}}{{2\left(\! {1 + {h_{{m_2}}}{x_{{i_2}}}{\varsigma ^2}}\! \right){\sigma ^2}}}}\!\! \right]}}{{\sqrt {1 + {h_{{m_2}}}{x_{{i_2}}}{\varsigma ^2}} }}} } \! +\! \frac{{{2^{p + 1}} \!-\! M}}{{M - {2^p}}}\!\!\!\sum\limits_{{x_{{i_2}}} \in {\bf{B}}}\! {\sum\limits_{{m_2} \in \mathbf{\Xi} } \!\!{\frac{{\exp\!\! \left[\!\! { - \frac{{{{\left( {z + d_{m,i}^{{m_2},{i_2}}} \right)}^2}}}{{2\left(\! {1 + {h_{{m_2}}}{x_{{i_2}}}{\varsigma ^2}} \!\right){\sigma ^2}}}} \!\! \right]}}{{\sqrt {1 + {h_{{m_2}}}{x_{{i_2}}}{\varsigma ^2}} }}} } }}}\! }\! \right]} } \nonumber\\
&&\!\!\!\!\!\!\!\!\! + \frac{{{2^{p \!+\! 1}} \!\!-\!\! M}}{{{2^{2p + q}}}}\!\!\sum\limits_{m \in \mathbf{\Xi} }\! {\sum\limits_{{x_i} \in {\bf{B}}} {\!{\mathbb{E}_z}\!\!\!\left[\! {{{\log }_2}\!\! {\frac{{\frac{{\exp \left[ { - \frac{{{z^2}}}{{2(1 + {h_m}{x_i}{\varsigma ^2}){\sigma ^2}}}} \right]}}{{2\sqrt {1 + {h_m}{x_i}{\varsigma ^2}} }}}}{{\frac{{M - {2^p}}}{{{2^{p \!+\!1}} \!-\! M}}\!\!\!\sum\limits_{{x_{{i_2}}} \in {\bf{A}}} \! {\sum\limits_{{m_2} \in \mathbf{\Psi}  \cup \mathbf{\Phi} }\!\! {\frac{{\exp \!\left[\! { - \frac{{{{\left( {z + d_{m,i}^{{m_2},{i_2}}} \right)}^2}}}{{2\left(\! {1 \!+\! {h_{{m_2}}}{x_{{i_2}}}{\varsigma ^2}} \!\right){\sigma ^2}}}}\! \right]}}{{\sqrt {1 + {h_{{m_2}}}{x_{{i_2}}}{\varsigma ^2}} }}} }  \!\!+\!\!\! \sum\limits_{{x_{{i_2}}} \in {\bf{B}}}\!{\sum\limits_{{m_2} \in \mathbf{\Xi} }\!\! {\frac{{\exp\!\left[\! { - \frac{{{{\left(\! {z + d_{m,i}^{{m_2},{i_2}}}\! \right)}^2}}}{{2\left(\! {1 \!+\! {h_{{m_2}}}{x_{{i_2}}}{\varsigma ^2}} \! \right){\sigma ^2}}}} \right]}}{{\sqrt {1 + {h_{{m_2}}}{x_{{i_2}}}{\varsigma ^2}} }}} } }}} } \!\!\right]} }.
 \label{eq23}
\end{eqnarray}
\hrulefill
\end{table*}
where $d_{m,i}^{{m_2},{i_2}} = {h_m}{x_i} - {h_{{m_2}}}{x_{{i_2}}}$, and $z \sim {\cal N}(0,(1 + {h_m}{x_i}{\varsigma ^2}){\sigma ^2})$ holds when given $h_m$ and $x_i$.
\label{the1}
\end{theorem}

\begin{proof}
See Appendix \ref{appa}.
\end{proof}

\begin{remark}
Define the SNR as $\gamma  = P_t/{\sigma ^2}$.
It can be easily proved that ${\cal I}\left( {x,h;y} \right)$ in (\ref{eq23}) is a monotonic increasing function with $\gamma $.
Moreover, when $\gamma  \to \infty $, we have
\begin{eqnarray}
\mathop {\lim }\limits_{\gamma  \to \infty } {\cal I}\left( {x,h;y} \right)
\!\!\!\! &=&\!\!\!\! \frac{{{{\left( {M - {2^p}} \right)}^2}}\left( {p + 1} \right)}{{{2^{2p}}}} \nonumber\\
&+&\!\!\!\! \frac{{{{\left( {{2^{p + 1}} - M} \right)}^2}}p}{{{2^{2p}}}} + \frac{{M - {2^p}}}{{{2^p}}}{\log _2}\frac{{{2^{p + q - 1}}}}{{M - {2^p}}}\nonumber\\
&+&\!\!\!\! \frac{{{2^{p + 1}} - M}}{{{2^p}}}{\log _2}\frac{{{2^{p + q}}}}{{{2^{p + 1}} - M}}.
 \label{eq24}
\end{eqnarray}
Moreover, when $\gamma  \to 0$, we can get (\ref{eq25}) as shown at the top of the next page.
\begin{table*}\normalsize
\begin{eqnarray}
 \mathop {\lim }\limits_{\gamma  \to 0} {\cal I}\left( {x,h;y} \right) \!\!\!\!&=&\!\!\!\! \frac{{{{\left( {M \!-\! {2^p}} \right)}^2}}(p + 1)}{{{2^{2p}}}} \!+\! \frac{{{{\left( {{2^{p + 1}} \!-\! M} \right)}^2}}p}{{{2^{2p}}}} \!+\! \frac{{M \!-\! {2^p}}}{{{2^p}}}{\log _2}\frac{{{2^{p + q - 1}}}}{{M \!-\! {2^p}}} \!+\! \frac{{{2^{p + 1}} \!-\! M}}{{{2^p}}}{\log _2}\frac{{{2^{p + q}}}}{{{2^{p + 1}} \!-\! M}}\nonumber \\
 &-&\!\!\!\! \frac{{M \!-\! {2^p}}}{{{2^{2p + q}}}}\!\!\sum\limits_{m \in \mathbf{\Psi}  \cup \mathbf{\Phi} }\! {\sum\limits_{{x_i} \in {\bf{A}}} \!\! {{\mathbb{E}_z}\!\!\!\left[\! {{{\log }_2}\!\!\left(\! {\sum\limits_{{x_{{i_2}}} \in {\bf{A}}}\! {\sum\limits_{{m_2} \in \mathbf{\Psi}  \cup \mathbf{\Phi} }\!\! {\frac{{\sqrt {1 \!+\! {h_m}{x_i}{\varsigma ^2}} }}{{\sqrt {1 \!+\! {h_{{m_2}}}{x_{{i_2}}}{\varsigma ^2}} }}} }  \!+\! \frac{{{2^{p + 1}} \!-\! M}}{{M \!-\! {2^p}}}\!\!\sum\limits_{{x_{{i_2}}} \in {\bf{B}}}\! {\sum\limits_{{m_2} \in \mathbf{\Xi} }\!\! {\frac{{\sqrt {1 \!+\! {h_m}{x_i}{\varsigma ^2}} }}{{\sqrt {1 \!+\! {h_{{m_2}}}{x_{{i_2}}}{\varsigma ^2}} }}} } }\!\! \right)}\!\! \right]} } \nonumber\\
 &-&\!\!\!\!\frac{{{2^{p + 1}} \!-\! M}}{{{2^{2p + q}}}}\!\!\sum\limits_{m \in \mathbf{\Xi} }\! {\sum\limits_{{x_i} \in {\bf{B}}}\!\! {{\mathbb{E}_z}\!\!\!\left[\!\! {{{\log }_2}\!\!\left(\!\! {\frac{{M - {2^p}}}{{{2^{p + 1}} \!-\! M}}\!\!\sum\limits_{{x_{{i_2}}} \in {\bf{A}}}\! {\sum\limits_{{m_2} \in \mathbf{\Psi}  \cup \mathbf{\Phi} } \!\!{\frac{{\sqrt {1 \!+\! {h_m}{x_i}{\varsigma ^2}} }}{{\sqrt {1 \!+\! {h_{{m_2}}}{x_{{i_2}}}{\varsigma ^2}} }}} }  \!+\!\!\! \sum\limits_{{x_{{i_2}}} \in {\bf{B}}}\! {\sum\limits_{{m_2} \in \mathbf{\Xi} } \!\! {\frac{{\sqrt {1 \!+\! {h_m}{x_i}{\varsigma ^2}} }}{{\sqrt {1 \!+\! {h_{{m_2}}}{x_{{i_2}}}{\varsigma ^2}} }}} } }\!\! \right)}\!\! \right]} }.
 \label{eq25}
\end{eqnarray}
\hrulefill
\end{table*}
\end{remark}

\begin{remark}
When $M$ is a power of 2 (i.e., $M = {2^p}$) and $N = {2^q}$, the mutual information ${\cal I}\left( {x,h;y} \right)$ in (\ref{eq23}) becomes
\begin{eqnarray}
\!\!\!\!{\cal I}\left( {x,h;y} \right)\!\!\!\!\!\!\!  &=&\!\!\!\!\! {\log _2}(MN) - \frac{1}{{MN}}\sum\limits_{m = 1}^M \sum\limits_{i = 1}^N \nonumber\\
&&\!\!\!\!\!{{\mathbb{E}_z}\left[ {{{\log }_2}\sum\limits_{{i_2} = 1}^N {\sum\limits_{{m_2} = 1}^M {\frac{{\sqrt {1 + {h_m}{x_i}{\varsigma ^2}} }}{{\sqrt {1 + {h_{{m_2}}}{x_{{i_2}}}{\varsigma ^2}} }}} } } \right.}  \nonumber\\
&\times&\!\!\!\!\!\!\!\left. {  \exp\!\!\! \left(\! {\frac{{{z^2}}}{{2(1 \!\!+\!\! {h_m}{x_i}{\varsigma ^2}){\sigma ^2}}} \!-\! \frac{{{{(z + d_{m,i}^{{m_2},{i_2}})}^2}}}{{2(1 \!\!+\!\! {h_{{m_2}}}{x_{{i_2}}}{\varsigma ^2}){\sigma ^2}}}}\!\! \right)}\!\!\! \right]\!\!.
\label{eq27}
\end{eqnarray}
Furthermore, when $M = {2^p}$ and $N = {2^q}$, (\ref{eq24}) and (\ref{eq25}) become
\begin{equation}
\mathop {\lim }\limits_{\gamma  \to \infty } {\cal I}\left( {x,h;y} \right) = {\log _2}(MN),
\label{eq28}
\end{equation}
and
\begin{eqnarray}
&&\mathop {\lim }\limits_{\gamma  \to 0} {\cal I}\left( {x,h;y} \right) = {\log _2}(MN) - \frac{1}{{MN}}\nonumber\\
&&\;\;\;\;\;\times \sum\limits_{m = 1}^M {\sum\limits_{i = 1}^N {{{\log }_2}\!\!\left( {\sum\limits_{{m_2} = 1}^M {\sum\limits_{{i_2} = 1}^N \!\!{\frac{{\sqrt {1 \!+\! {h_m}{x_i}{\varsigma ^2}} }}{{\sqrt {1 \!+\! {h_{{m_2}}}{x_{{i_2}}}{\varsigma ^2}} }}} } }\!\! \right)} }.
\label{eq29}
\end{eqnarray}
Note that the results in (\ref{eq27})-(\ref{eq29}) are consistent with the conclusions in \cite{BIB06}.
\end{remark}

\subsection{Lower bound on Mutual Information}
\label{section4_2}
In (\ref{eq23}), it is very difficult to derive a closed-form expression for the mutual information.
To reduce computational complexity and obtain more insights,
a lower bound on the mutual information is derived in the following theorem.

\begin{theorem}
For the SM based VLC, a closed-form expression for the lower bound of the mutual information is given by (\ref{eq34}) as shown at the top of the next page,
\begin{table*}\normalsize
\begin{eqnarray}
{{\cal I}_{{\rm{Low}}}}\!\left( {x,h;y} \right) \!\!&=&\!\! \frac{{{{( {M \!-\! {2^p}})}^2} \!\!+\!\! {{\left( {{2^{p \!+\! 1}} \!\!-\!\! M} \right)}^2}}}{{{2^{2p}}}}\!\!\left(\!\! {p \!+\! 1 \!-\! \frac{{{{\log }_2}e}}{2}}\! \!\right) \!\!+\! \frac{{M \!-\! {2^p}}}{{{2^p}}}{\log _2}\frac{{{2^{p \!+\! q \!-\! 1}}}}{{M \!\!-\!\! {2^p}}} \!+\! \frac{{{2^{p + 1}} \!-\! M}}{{{2^p}}}{\log _2}\frac{{{2^{p + q}}}}{{{2^{p + 1}} \!-\! M}} \nonumber \\
&-& \frac{{M - {2^p}}}{{{2^{2p + q}}}}\!\!\sum\limits_{m \in \mathbf{\Psi}  \cup \mathbf{\Phi }} {\sum\limits_{{x_i} \in {\bf{A}}} {{{\log }_2}\!\!\left[ {\sum\limits_{{x_{{i_2}}} \in \mathbf{A}} {\sum\limits_{{m_2} \in \mathbf{\Psi}  \cup \mathbf{\Phi} } {\frac{{\sqrt {1 \!+\! {h_m}{x_i}{\varsigma ^2}} }}{{\sqrt {2\left( {1 \!+\! {h_{{m_2}}}{x_{{i_2}}}{\varsigma ^2}} \right)} }}\exp\!\left(\!\! { - \frac{{{{\left( {d_{m,i}^{{m_2},{i_2}}} \right)}^2}}}{{4\left( {1 \!+\! {h_{{m_2}}}{x_{{i_2}}}{\varsigma ^2}} \right){\sigma ^2}}}}\!\! \right)} } } \right.} } \nonumber\\
&+& \left. {  \frac{{{2^{p + 1}} - M}}{{M - {2^p}}}\sum\limits_{{x_{{i_2}}} \in \mathbf{B}} {\sum\limits_{{m_2} \in \mathbf{\Xi} } {\frac{{\sqrt {1 + {h_m}{x_i}{\varsigma ^2}} }}{{\sqrt {2\left( {1 + {h_{{m_2}}}{x_{{i_2}}}{\varsigma ^2}} \right)} }}} \exp\left( { - \frac{{{{\left( {d_{m,i}^{{m_2},{i_2}}} \right)}^2}}}{{4\left( {1 + {h_{{m_2}}}{x_{{i_2}}}{\varsigma ^2}} \right){\sigma ^2}}}} \right)} } \right]\nonumber\\
 &-& \frac{{{2^{p \!+\! 1}} \!-\! M}}{{{2^{2p + q}}}}\!\!\sum\limits_{m \in \mathbf{\Xi} }\! {\sum\limits_{{x_i} \in {\bf{B}}}\!\! {{{\log }_2}\!\!\!\left[\!\! {\frac{{M \!-\! {2^p}}}{{{2^{p \!+\! 1}} \!-\! M}}\!\!\sum\limits_{{x_{{i_2}}} \in {\bf{A}}}\! {\sum\limits_{{m_2} \in \mathbf{\Psi}  \cup \mathbf{\Phi} }\!\!\!\! {\frac{{\sqrt {2\!\left( {1 \!+\! {h_m}{x_i}{\varsigma ^2}} \right)} }}{{\sqrt {1 \!+\! {h_{{m_2}}}{x_{{i_2}}}{\varsigma ^2}} }}\exp\!\!\left(\!\! { - \frac{{{{\left( {d_{m,i}^{{m_2},{i_2}}} \right)}^2}}}{{4\left( {1 \!+\! {h_{{m_2}}}{x_{{i_2}}}{\varsigma ^2}} \right){\sigma ^2}}}}\!\! \right)} } } \right.} }\nonumber \\
&+& \left. { \sum\limits_{{x_{{i_2}}} \in {\bf{B}}} {\sum\limits_{{m_2} \in \mathbf{\Xi} } {\frac{{\sqrt {2\left( {1 + {h_m}{x_i}{\varsigma ^2}} \right)} }}{{\sqrt {1 + {h_{{m_2}}}{x_{{i_2}}}{\varsigma ^2}} }}} \exp\left( { - \frac{{{{\left( {d_{m,i}^{{m_2},{i_2}}} \right)}^2}}}{{4\left( {1 + {h_{{m_2}}}{x_{{i_2}}}{\varsigma ^2}} \right){\sigma ^2}}}} \right)} } \right].
 \label{eq34}
\end{eqnarray}
\hrulefill
\end{table*}
where $d_{m,i}^{{m_2},{i_2}} = {h_m}{x_i} - {h_{{m_2}}}{x_{{i_2}}}$.
\label{the2}
\end{theorem}

\begin{proof}
See Appendix \ref{appb}.
\end{proof}

\begin{remark}
It can be easily proved that ${{\cal I}_{{\rm{Low}}}}\left( {x,h;y} \right)$ in (\ref{eq34}) is a monotonic increasing function with $\gamma $. Moreover, when $\gamma  \to \infty $, (\ref{eq34}) becomes
\begin{eqnarray}
&&\!\!\!\!\!\!\!\!\!\!\!\!\!\!\!\!\!\!\!\!\mathop {\lim }\limits_{\gamma  \to \infty } {{\cal I}_{{\rm{Low}}}}\left( {x,h;y} \right)=\frac{{{{\left( {M - {2^p}} \right)}^2}}}{{{2^{2p}}}}\left( {p + \frac{3}{2} - \frac{1}{2}{{\log }_2}e} \right) \nonumber\\
&&\!\!\!\!\!\!\!\!\!\!\!\!\!+ \frac{{{{\left( {{2^{p + 1}} - M} \right)}^2}}}{{{2^{2p}}}}\left( {p + \frac{1}{2} - \frac{1}{2}{{\log }_2}e} \right)\nonumber\\
&&\!\!\!\!\!\!\!\!\!\!\!\!\!+ \frac{{M - {2^p}}}{{{2^p}}}{\log _2}\frac{{{2^{p + q - 1}}}}{{M - {2^p}}} + \frac{{{2^{p + 1}} - M}}{{{2^p}}}{\log _2}\frac{{{2^{p + q}}}}{{{2^{p + 1}} - M}},
\label{eq35}
\end{eqnarray}
and when $\gamma  \to 0$, (\ref{eq34}) becomes (\ref{eq36}) as shown at the top of the next page.
\begin{table*}\normalsize
\begin{eqnarray}
&&\!\!\!\!\!\!\!\! \mathop {\lim }\limits_{\gamma  \to 0} {{\cal I}_{{\rm{Low}}}}\left( {x,h;y} \right) \!=\! \frac{{{{\left( {M \!-\! {2^p}} \right)}^2}}}{{{2^{2p}}}}\left(\! {p \!+\! \frac{3}{2} \!-\! \frac{1}{2}{{\log }_2}e} \!\right) \!+\! \frac{{{{\left( {{2^{p + 1}} \!-\! M} \right)}^2}}}{{{2^{2p}}}}\left(\! {p \!+\! \frac{1}{2} \!-\! \frac{1}{2}{{\log }_2}e}\! \right)+ \frac{{M - {2^p}}}{{{2^p}}}{\log _2}\frac{{{2^{p + q - 1}}}}{{M - {2^p}}}\nonumber\\
&&\!\!\!\!\!\!\!\!  + \frac{{{2^{p + 1}} \!-\! M}}{{{2^p}}}{\log _2}\frac{{{2^{p + q}}}}{{{2^{p + 1}} \!-\! M}} \!-\! \frac{{M \!-\! {2^p}}}{{{2^{2p + q}}}}\!\!\!\!\sum\limits_{m \in \mathbf{\Psi}  \cup \mathbf{\Phi} } \!{\sum\limits_{{x_i} \in \mathbf{A}}\!\! {{{\log }_2}\!\!\left(\! {\sum\limits_{{x_{{i_2}}} \in \mathbf{A}}\! {\sum\limits_{{m_2} \in \mathbf{\Psi}  \cup \mathbf{\Phi} }\!\! {\frac{{\sqrt {1 \!+\! {h_m}{x_i}{\varsigma ^2}} }}{{\sqrt {1 \!+\! {h_{{m_2}}}{x_{{i_2}}}{\varsigma ^2}} }}} }  \!+\! \frac{{{2^{p + 1}} \!-\! M}}{{M \!-\! {2^p}}}\sum\limits_{{x_{{i_2}}} \in \mathbf{B}}\! {\sum\limits_{{m_2} \in \mathbf{\Xi} }\!\! {\frac{{\sqrt {1 \!+\! {h_m}{x_i}{\varsigma ^2}} }}{{\sqrt {1 \!+\! {h_{{m_2}}}{x_{{i_2}}}{\varsigma ^2}} }}} } } \!\!\right)} } \nonumber\\
&&\!\!\!\!\!\!\!\!- \frac{{{2^{p + 1}} \!-\! M}}{{{2^{2p + q}}}}\!\!\sum\limits_{m \in \mathbf{\Xi} }\! {\sum\limits_{{x_i} \in \mathbf{B}}\!\! {{{\log }_2}\!\!\!\left(\!\! {\frac{{M \!-\! {2^p}}}{{{2^{p \!+\! 1}} \!-\! M}}\!\!\!\sum\limits_{{x_{{i_2}}} \in \mathbf{A}}\! {\sum\limits_{{m_2} \in \mathbf{\Psi}  \cup \mathbf{\Phi} }\!\!\! {\frac{{\sqrt {1 \!+\! {h_m}{x_i}{\varsigma ^2}} }}{{\sqrt {1 \!+\! {h_{{m_2}}}{x_{{i_2}}}{\varsigma ^2}} }}} }  \!+\!\!\! \sum\limits_{{x_{{i_2}}} \in \mathbf{B}}\! {\sum\limits_{{m_2} \in \mathbf{\Xi} }\!\! {\frac{{\sqrt {1 \!+\! {h_m}{x_i}{\varsigma ^2}} }}{{\sqrt {1 \!+\! {h_{{m_2}}}{x_{{i_2}}}{\varsigma ^2}} }}} } }\!\! \right)} }.
\label{eq36}
\end{eqnarray}
\hrulefill
\end{table*}
\end{remark}

\begin{remark}
When $\gamma  \to \infty $ and $\gamma  \to 0$, the performance gap between (\ref{eq35}) and (\ref{eq36}) can be written, respectively, as
\begin{eqnarray}
&&\!\!\!\!\!\!\!\!\!\!\!\!\!\!\!\!\!\!\!\!\mathop {\lim }\limits_{\gamma  \to \infty } \left[ {{\cal I}\left( {x,h;y} \right) - {{\cal I}_{{\rm{Low}}}}\left( {x,h;y} \right)} \right] \nonumber\\
&&= \frac{{{{\left( {M - {2^p}} \right)}^2} + {{\left( {{2^{p + 1}} - M} \right)}^2}}}{{{2^{2p + 1}}}}\left( {\log_2 e  - 1} \right),
\label{eq37}
\end{eqnarray}
and
\begin{eqnarray}
&&\!\!\!\!\!\!\!\!\!\!\!\!\!\!\!\!\!\!\!\!\mathop {\lim }\limits_{\gamma  \to 0} \left[ {{\cal I}\left( {x,h;y} \right) - {{\cal I}_{{\rm{Low}}}}\left( {x,h;y} \right)} \right] \nonumber\\
&&= \frac{{{{\left( {M - {2^p}} \right)}^2} + {{\left( {{2^{p + 1}} - M} \right)}^2}}}{{{2^{2p + 1}}}}\left( {\log_2 e - 1} \right).
\label{eq38}
\end{eqnarray}
From (\ref{eq37}) and (\ref{eq38}), it can be concluded that, in the low and high SNR regimes, a constant performance gap exists between ${\cal I}\left( {x,h;y} \right)$ and ${{\cal I}_{{\rm{Low}}}}\left( {x,h;y} \right)$.
\label{rmk4}
\end{remark}

\begin{remark}
When $M$ is a power of 2 (i.e., $M = {2^p}$) and $N = {2^q}$, the lower bound of the mutual information ${{\cal I}_{{\rm{Low}}}}\left( {x,h;y} \right)$ in (\ref{eq34}) becomes
\begin{eqnarray}
&&\!\!\!\!\!\!\!\!\!\!{{\cal I}_{{\rm{Low}}}}\left( {x,h;y} \right) = {\log _2}MN + \frac{1}{2}\left( {1 - {{\log }_2}e} \right) - \frac{1}{{MN}}\sum\limits_{m = 1}^M \sum\limits_{i = 1}^N \nonumber\\
&&\!\!\!\!\!\!\!\!\!\! {{{\log }_2}\!\!\left[\! {\sum\limits_{{m_2} = 1}^M\! {\sum\limits_{{i_2} = 1}^N \!\!{\frac{{\sqrt {1 \!+\! {h_m}{x_i}{\varsigma ^2}} }}{{\sqrt {1 \!+\! {h_{{m_2}}}{x_{{i_2}}}{\varsigma ^2}} }}\exp\!\!\left[\!\! {  \frac{{{{-\left( {d_{m,i}^{{m_2},{i_2}}} \right)}^2}}}{{4\left( {1 \!+\! {h_{{m_2}}}\!{x_{{i_2}}}\!{\varsigma ^2}} \right)\!{\sigma ^2}}}}\!\! \right]} } }\! \right]}.
 \label{eq40}
\end{eqnarray}
Furthermore, when $M = {2^p}$ and $N = {2^q}$, (\ref{eq35}) and (\ref{eq36}) become
\begin{eqnarray}
\mathop {\lim }\limits_{\gamma  \to \infty } {{\cal I}_{{\rm{Low}}}}\left( {x,h;y} \right) = {\log _2}\left( {MN} \right) - \frac{1}{2}\left( {\log_2 e - 1} \right),
\label{eq41}
\end{eqnarray}
and
\begin{eqnarray}
&&\!\!\!\!\!\!\!\!\!\!\!\!\!\!\!\!\!\!\!\! \mathop {\lim }\limits_{\gamma  \to 0} {{\cal I}_{{\rm{Low}}}}\left( {x,h;y} \right) \!=\! {\log _2}\left( {MN} \right) \!-\! \frac{1}{2}\left( {\log_2 e \!-\! 1} \right) \nonumber\\
&&\!-\! \frac{1}{{MN}}\! \sum\limits_{m = 1}^M \!{\sum\limits_{i = 1}^N \!{{{\log }_2}\!\!\left( {\sum\limits_{{m_2} = 1}^M \!{\sum\limits_{{i_2} = 1}^N \!{\frac{{\sqrt {1 \!+\! {h_m}{x_i}{\varsigma ^2}} }}{{\sqrt {1 \!+\! {h_{{m_2}}}{x_{{i_2}}}{\varsigma ^2}} }}} } } \!\!\right)} }.
\label{eq42}
\end{eqnarray}
Furthermore, when $M = {2^p}$ and $N = {2^q}$, the gap between (\ref{eq28}) and (\ref{eq41}) is given by
\begin{eqnarray}
\mathop {\lim }\limits_{\gamma  \to \infty } \left[ {{\cal I}\left( {x,h;y} \right) - {{\cal I}_{\rm Low}}\left( {x,h;y} \right)} \right] = \frac{1}{2}\left( {\log_2} e - 1 \right).
\label{eq43}
\end{eqnarray}
Similarly, the gap between (\ref{eq29}) and (\ref{eq42}) is given by
\begin{eqnarray}
\mathop {\lim }\limits_{\gamma  \to 0} \left[ {{\cal I}\left( {x,h;y} \right) - {{\cal I}_{\rm Low}}\left( {x,h;y} \right)} \right] = \frac{1}{2}\left( {\log_2} e - 1 \right).
\label{eq44}
\end{eqnarray}
It can be seen from (\ref{eq43}) and (\ref{eq44}) that, when $M$ is a power of 2, the performance gap between ${\cal I}\left( {x,h;y} \right)$ and ${{\cal I}_{{\rm{Low}}}}\left( {x,h;y} \right)$ is a constant for both low and high SNR regimes, which is the same as that in \cite{BIB06}. Moreover, (\ref{eq43}) and (\ref{eq44}) can also be derived by letting $M = {2^p}$ in \emph{Remark \ref{rmk4}}.
\end{remark}

\section{Precoding Design}
\label{section5}
IT can be known from (\ref{eq28}) that the upper bound of the mutual information is $\log_2{MN}$.
However, such an upper bound may not be achievable in some cases.
To further improve the system performance, a precoding scheme is proposed in this section. By employing the precoding, the $i$-th symbol transmitted by the $m$-th LED is multiplied by a coefficient ${w_{m,i}}$, so the precoding process can be regarded as a mapping
\begin{eqnarray}
\eta :{\bf{H}} \times {\bf{X}} \to {\bf{R}},\;{\rm{where}}\;{\bf{R}} = \left\{ {{w_{m,i}}{h_m}{x_i}|\forall m,i} \right\},
\label{eq45}
\end{eqnarray}
where ${\bf{R}}$ denotes the received signal space for non-noise channel.

When using precoding, the lower bound of the mutual information (\ref{eq34}) can be reformulated as (\ref{eq46}) as shown at the top of the next page,
\begin{table*}\normalsize
\begin{eqnarray}
&&\!\!\!\!\! {{\cal I}_{\rm Low}}\!( {x,h;y}) \!\!=\!\! \frac{{{{\left( {M \!\!-\!\! {2^p}} \right)}^2} \!\!+\!\! {{\left( {{2^{p + 1}} \!\!-\!\! M} \right)}^2}}}{{{2^{2p}}}}\!\!\left(\!\! {p \!+\! 1 \!-\! \frac{{{{\log }_2}e}}{2}}\!\! \right) \!\!+\!\! \frac{{M \!\!-\!\! {2^p}}}{{{2^p}}}{\log _2}\frac{{{2^{p + q - 1}}}}{{M \!\!-\!\! {2^p}}} \!\!+\!\! \frac{{{2^{p + 1}} \!\!-\!\! M}}{{{2^p}}}{\log _2}\frac{{{2^{p + q}}}}{{{2^{p \!+\! 1}}\!\! -\!\! M}} \nonumber \\
&&\!\!\!\!\! - \frac{{M \!-\! {2^p}}}{{{2^{2p + q}}}}\!\!\sum\limits_{m \in \mathbf{\Psi}  \cup \mathbf{\Phi} }\! {\sum\limits_{{x_i} \in {\bf{A}}} {{\!{\log }_2}\!\!\left[\! {\sum\limits_{{x_{{i_2}}} \in \mathbf{A}}\! {\sum\limits_{{m_2} \in \mathbf{\Psi}  \cup \mathbf{\Phi} }\!\! {\frac{{\sqrt {1 \!+\! {w_{m,i}}{h_m}{x_i}{\varsigma ^2}} }}{{\sqrt {2\left( {1 \!+\! {w_{{m_2},{i_2}}}{h_{{m_2}}}{x_{{i_2}}}{\varsigma ^2}} \right)} }}\exp \!\left(\!\! { \frac{{{{-\left( {d_{m,i}^{{m_2},{i_2}}} \right)}^2}}}{{4\left( {1 \!+\! {w_{{m_2},{i_2}}}{h_{{m_2}}}{x_{{i_2}}}{\varsigma ^2}} \right)\!{\sigma ^2}}}}\!\! \right)} } } \right.} } \nonumber\\
&&\!\!\!\!\!\left. { + \frac{{{2^{p + 1}} - M}}{{M - {2^p}}}\sum\limits_{{x_{{i_2}}} \in {\bf{B}}} {\sum\limits_{{m_2} \in \mathbf{\Xi} } {\frac{{\sqrt {1 + {w_{m,i}}{h_m}{x_i}{\varsigma ^2}} }}{{\sqrt {2\left( {1 + {w_{{m_2},{i_2}}}{h_{{m_2}}}{x_{{i_2}}}{\varsigma ^2}} \right)} }}} \exp \left( { - \frac{{{{\left( {d_{m,i}^{{m_2},{i_2}}} \right)}^2}}}{{4\left( {1 + {w_{{m_2},{i_2}}}{h_{{m_2}}}{x_{{i_2}}}{\varsigma ^2}} \right){\sigma ^2}}}} \right)} } \right] \nonumber \\
&&\!\!\!\!\! - \frac{{{2^{p \!+\! 1}} \!\!-\!\! M}}{{{2^{2p + q}}}}\!\!\sum\limits_{m \in \mathbf{\Xi} } \!{\sum\limits_{{x_i} \in {\bf{B}}} \!\!{{{\log }_2}\!\!\!\left[\! {\frac{{M \!-\! {2^p}}}{{{2^{p \!+\! 1}} \!\!-\!\! M}}\!\!\sum\limits_{{x_{{i_2}}} \in \mathbf{A}}\! {\sum\limits_{{m_2} \in \mathbf{\Psi}  \cup \mathbf{\Phi} }\!\! {\frac{{\sqrt {2\left( {1 \!+\! {w_{m,i}}{h_m}{x_i}{\varsigma ^2}} \right)} }}{{\sqrt {1 \!+\! {w_{{m_2},{i_2}}}{h_{{m_2}}}{x_{{i_2}}}{\varsigma ^2}} }}\!\exp\!\! \left(\!\! {  \frac{{{{-\left( {d_{m,i}^{{m_2},{i_2}}} \right)}^2}}}{{4\left( {1 \!+\! {w _{{m_2},{i_2}}}{h_{{m_2}}}{x_{{i_2}}}{\varsigma ^2}} \right)\!{\sigma ^2}}}}\!\! \right)} } } \right.} } \nonumber\\
&&\!\!\!\!\!\left. { + \sum\limits_{{x_{{i_2}}} \in {\bf{B}}} {\sum\limits_{{m_2} \in \Xi } {\frac{{\sqrt {2\left( {1 + {w_{m,i}}{h_m}{x_i}{\varsigma ^2}} \right)} }}{{\sqrt {1 + {w_{{m_2},{i_2}}}{h_{{m_2}}}{x_{{i_2}}}{\varsigma ^2}} }}} \exp \left( { - \frac{{{{\left( {d_{m,i}^{{m_2},{i_2}}} \right)}^2}}}{{4\left( {1 + {w_{{m_2},{i_2}}}{h_{{m_2}}}{x_{{i_2}}}{\varsigma ^2}} \right){\sigma ^2}}}} \right)} } \right].
\label{eq46}
\end{eqnarray}
\hrulefill
\end{table*}
where $d_{m,i}^{{m_2},{i_2}} = |{w_{m,i}}{h_m}{x_i} - {w_{{m_2},{i_2}}}{h_{{m_2}}}{x_{{i_2}}}|$ denotes the distance between any two points in the received signal space ${\bf{R}}$.

In (\ref{eq46}), the minimum distance $ \mathop {\min }\limits_{(m,i) \ne ({m_2},{i_2})} \left\{ {d_{m,i}^{{m_2},{i_2}}}/\sqrt{1+w_{m_2,i_2}h_{m_2} x_{i_2} \varsigma ^2} \right\}$ dominates the total performance of the system,
and thus the purpose of the precoding in this paper is to maximize the minimum distance.
To facilitate the comparison, the average optical intensity without precoding is set to be the same as that with precoding.
Mathematically, the optimization problem for precoding is formulated as
\begin{eqnarray}
    && \mathop {\max }\limits_{{w_{m,i}}} \mathop {\min }\limits_{(m,i) \ne ({m_2},{i_2})} \left\{ {d_{m,i}^{{m_2},{i_2}}}/\sqrt{1+w_{m_2,i_2}h_{m_2} x_{i_2} \varsigma ^2} \right\} \nonumber \\
    \text{s.t.} && {\sum\limits_{m \in \mathbf{\Psi}  \cup \mathbf{\Phi} }\! {\sum\limits_{{x_i} \in {\bf{A}}} {{w_{m,i}}{h_m}{x_i}} }  \!+\! \sum\limits_{m \in \mathbf{\Xi} }\! {\sum\limits_{{x_i} \in {\bf{B}}}\!\! {{w_{m,i}}{h_m}{x_i}} } } \nonumber\\
     &&\;\;\;\;\;\;=\!\!\! {\sum\limits_{m \in \mathbf{\Psi}  \cup \mathbf{\Phi} } \!{\sum\limits_{{x_i} \in {\bf{A}}} \!\!{{h_m}{x_i}} }  \!+\! \sum\limits_{m \in \mathbf{\Xi} }\! {\sum\limits_{{x_i} \in {\bf{B}}}\!\! {{h_m}{x_i}} } }.
\label{eq47}
\end{eqnarray}
Note that the optimization problem (\ref{eq47}) is non-convex and non-differentiable,
which is difficult to solve.
Let $L={d_{m,i}^{{m_2},{i_2}}}/\sqrt{1+w_{m_2,i_2}h_{m_2} x_{i_2} \varsigma ^2}$, a continuous approximation is provided for the inner minimum function \cite{BIB19}
\begin{eqnarray}
&&\!\!\!\!\!\!\!\!\!\!\!\!\mathop {\min }\limits_{(m,i) \ne ({m_2},{i_2})}\! \{ L\} \! \approx  - \frac{1}{\rho }\ln\!\! \left[\!\! {\sum\limits_{m \!\in\! \bf{\Psi}  \cup \bf{\Phi }}\! {\sum\limits_{{x_i} \!\in\! {\bf{A}}}\!\! {\left(\!\! {\sum\limits_{\scriptstyle{m_2} \in {\bf{\Psi}}  \cup {\bf{\Phi}} \hfill\atop
\scriptstyle{m_2} \ne m\hfill}\!\! {\sum\limits_{\scriptstyle{x_{{i_2}}} \!\in\! {\bf{A}}\hfill\atop
\scriptstyle{i_2} \ne i\hfill}\!\!\! {e^{\! - \rho L\!}} } } \right.} } } \right.\nonumber \\
&&\left. { + \sum\limits_{{m_2} \in \bf{\Xi} } {\sum\limits_{{x_{{i_2}}} \in {\bf{B}}} {e^{ - \rho L}} } } \right)\nonumber\\
&&\!\!\!\!\!\!\!\!\!\!\left. { +\! \sum\limits_{m \in \bf{\Xi} } {\sum\limits_{{x_i} \in {\bf{B}}}\! {\left(\! {\sum\limits_{{m_2} \in \bf{\Psi}  \cup \bf{\Phi} } {\sum\limits_{{x_{{i_2}}} \in {\bf{A}}}\!\! {e^ {\! - \rho L\!}} }  \!+\!\!\! \sum\limits_{\scriptstyle{m_2} \in {\bf{\Xi}} \hfill\atop
\scriptstyle{m_2} \ne m\hfill}\!\! {\sum\limits_{\scriptstyle{x_{{i_2}}} \in {\bf{B}}\hfill\atop
\scriptstyle{i_2} \ne i\hfill} \!\!{e^{\! - \rho L\!}} } }\! \right)} } } \right]\!\!,
\label{eq48}
\end{eqnarray}
where $\rho $ is positive and can be arbitrarily increased to improve model accuracy.

Submitting (\ref{eq48}) into (\ref{eq47}), both the objective and the constraint become differentiable,
and the problem can be solved by using the interior point method \cite{BIB20}.
However, the employed approximation (\ref{eq48}) may be not very good in some cases.
For example, when the interior point method approaches a local optima, the smallest distances between all points approach the same value.
This will cause (\ref{eq48}) to produce poor results for small value of $\rho $.
In this case, the value of $\rho $ should be enlarged to improve the accuracy of the approximation.
To facilitate the descriptions, the stepwise procedures of the interior point method based iteration algorithm are listed as Algorithm 2 in Fig. \ref{fig6}.

\begin{figure}[!h]
\hrulefill\\
\small
\textbf{Algorithm 2} (Interior point method based iteration algorithm)\vspace*{-5pt}\\
\vspace*{0pt}\hrulefill \\
\textbf{Step 1):} Given a feasible initial ${w_{m,i}},\forall m,i$. Initialize $\rho $ and ${\rho _{{\rm{stop}}}}$.\\
\textbf{Step 2):} Solve the problem by using the interior point method.\\
\textbf{Step 3):} \textbf{While} $\rho  < {\rho _{{\rm{stop}}}}$  \textbf{do}\\
                   \indent\qquad\qquad\qquad Let $\rho  \leftarrow 2\rho $;\\
                 \indent\qquad \qquad\qquad  Solve the problem by using the interior point method;\\
                \indent\qquad\qquad    \textbf{EndWhile}\\
\textbf{Step 4):} Output the derived ${w_{m,i}},\forall m,i$. \\
\vspace*{2pt} \hrulefill
\caption{Interior point method based iteration algorithm.}
\label{fig6}
\end{figure}

Furthermore, the computational complexity of the proposed algorithm is analyzed.
As is known, the inner-point method is a computational efficient algorithm \cite{BIB20}.
Assume that the complexity of the inner-point method for each iteration is $O(T)$.
The number of iterations for updating $\rho$ is upper bounded by $\left\lceil {{{\log }_2}({{{\rho _{{\rm{stop}}}}} \mathord{\left/
 {\vphantom {{{\rho _{{\rm{stop}}}}} \rho }} \right.
 \kern-\nulldelimiterspace} \rho })} \right\rceil {\rm{ + 1}}$, where $\left\lceil  \cdot  \right\rceil $ denotes the ceiling function.
Therefore, the total complexity of the proposed algorithm is on the order of $O(T\left\lceil {{{\log }_2}({{{\rho _{{\rm{stop}}}}} \mathord{\left/
 {\vphantom {{{\rho _{{\rm{stop}}}}} \rho }} \right.
 \kern-\nulldelimiterspace} \rho })} \right\rceil )$, which can be solved within polynomial time.

\section{Numerical Results}
\label{section6}
In this section, an indoor VLC system with multiple LEDs and one PD is considered,
which is deployed in a $5{\rm m} \times 4{\rm m} \times 3{\rm m}$ room.
The LEDs are installed on the ceiling of the room, while the PD is placed on a receiver plane with a height 0.8 m.
The other simulation parameters are listed in Table \ref{tab1}.

\begin{table}[!h]
\caption{Main simulation parameters.}
\label{tab1}
\centering
\begin{tabular}{|c|c|c|}
\hline\hline
\textbf{Parameter} & \textbf{Symbol} &\textbf{Value} \\
\hline
Physical area of the PD & $E$ & $1 {\rm cm}^2$ \\
\hline
Semi-angle at half-power of the LED & $\Phi_{1/2}$ & $35^0$ \\
\hline
Noise variance & $\sigma^2$ & $-104$ dBm \\
\hline
FOV of the PD & $\Psi _{\rm{c}}$ & $72^0$ \\
\hline\hline
\end{tabular}
\end{table}

\subsection{BER Results}
\label{section6_1}
In this subsection, the BER performance of the proposed CABM scheme will be shown.
To facilitate the comparison, the MV-IGCH scheme in \cite{BIB15} is also shown as a benchmark.

\begin{figure}
\centering
\includegraphics[width=8.5cm]{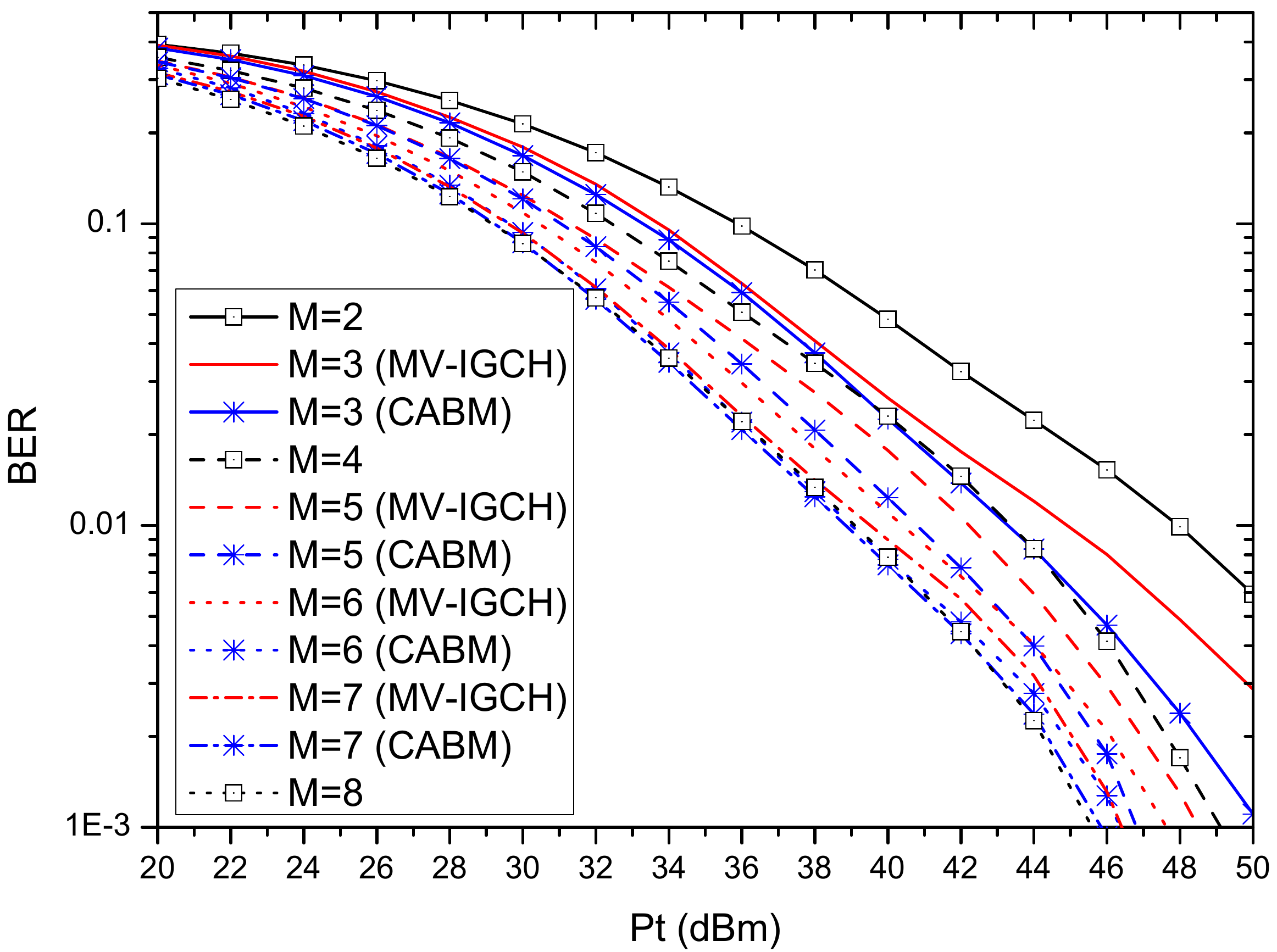}
\caption{BER comparisons between the proposed CABM scheme and the MV-IGCH scheme when $M$ varies from 2 to 8 and $K=5$ bits.}
\label{fig7}
\end{figure}

Fig. \ref{fig7} shows BER comparisons between the proposed CABM scheme and the MV-IGCH scheme
when $M$ varies from 2 to 8 and $K=5$ bits.
In this simulation, when $M=2$, $4$ and $8$, the conventional bit mapping scheme is utilized,
and 16-ary PAM and 8-ary PAM and $4$-ary PAM are used, respectively.
When $M=3$, 16-ary PAM and 8-ary PAM are randomly switched for signal domain symbols.
When $M = 5$, $6$ and $7$, 8-ary PAM and 4-ary PAM are randomly employed for each case.
Obviously, the values of BER for all curves decrease rapidly with the increase of $P_t$.
Moreover, generally the BER performance improves with the increase of $M$.
For large $M$, more information can be transmitted in space domain,
and lower-order constellations can be used in signal domain to reach the same performance as higher-order constellations used with a small $M$.
This conclusion coincides with that in \cite{BIB15}.
Moreover, the performance of the proposed CABM scheme always outperforms that of the MV-IGCH scheme, which verifies
the efficiency of the proposed scheme.

The total bit number of source data (i.e., $K$ bits) can be regarded as the transmission rate per timeslot.
To demonstrate the transmission rate on system performance,
Fig. \ref{fig8} shows BER comparisons between the proposed CABM scheme and the MV-IGCH scheme when $K=5, 6, 7$ bits and $M=5$.
As is observed, the BER performance improves with $P_t$, which is consist with that in Fig. \ref{fig7}.
Moreover, for a fixed number of LEDs, the BER performance degrades with the increase of transmission rate $K$.
This indicates there is a tradeoff between efficiency and reliability for information transmission.
Once again, it can be seen that the performance of the proposed CABM scheme is always better than that of the MV-IGCH scheme.
Specifically, when BER is $10^{-2}$, the proposed CABM scheme gives about 1 dB performance gain compared to the MV-IGCH scheme.
Moreover, the smaller the BER is, the larger the performance gain becomes.
This validates the advantage of the channel adaptive scheme.

\begin{figure}
\centering
\includegraphics[width=8.5cm]{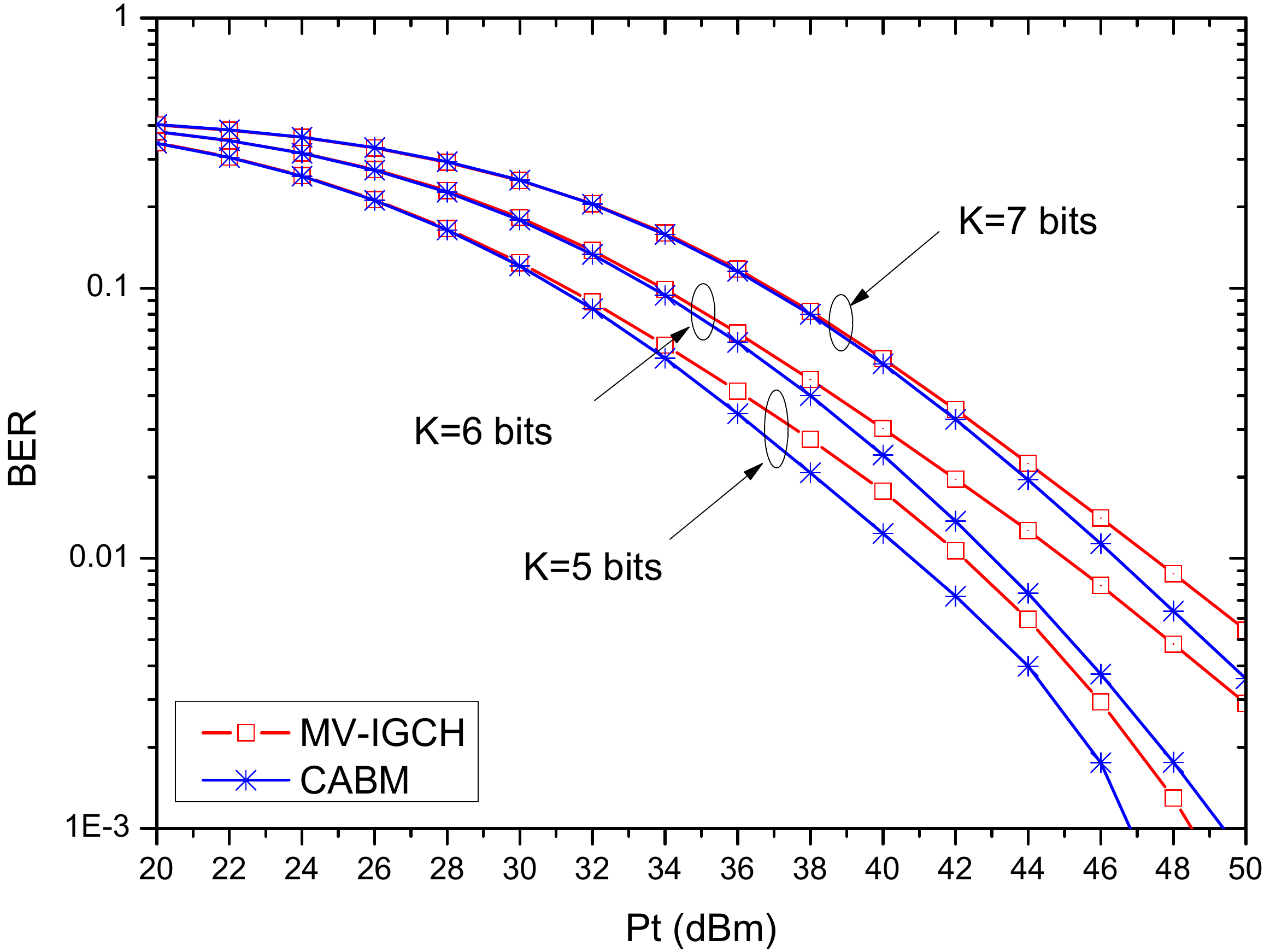}
\caption{BER comparisons between the proposed CABM scheme and the MV-IGCH scheme when $K=5, 6, 7$ bits and $M=5$.}
\label{fig8}
\end{figure}

To further demonstrate the effect of input-dependent noise on system performance,
Fig. \ref{fig9} shows BER comparisons between the proposed CABM scheme and the MV-IGCH scheme when $\varsigma = 0, 50, 100$, $M=3$ and $K=5$ bits.
Similar to Figs. \ref{fig7} and \ref{fig8}, the BER performance improves with $P_t$.
Moreover, in the low SNR regime, the difference among curves is not apparent.
However, at high SNR, the performance gap becomes large among different curves.
It can be seen that the BER performance degrades with the increase of $\varsigma$.
This indicates that the input-dependent noise has a strong impact on system performance at high SNR.
Furthermore, compared with the MV-IGCH scheme, the proposed CABM scheme can always achieve the better performance.

\begin{figure}
\centering
\includegraphics[width=8.5cm]{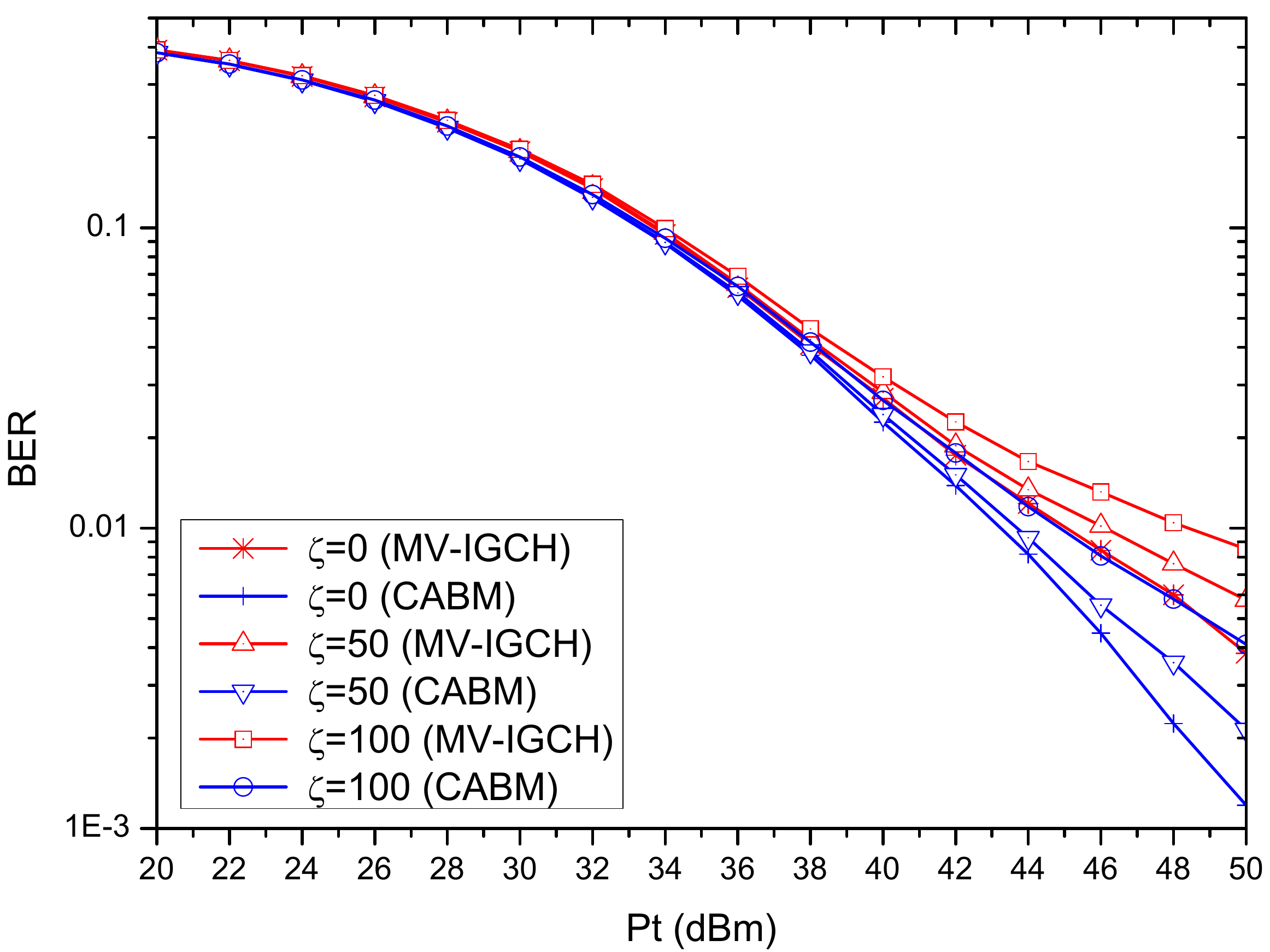}
\caption{BER comparisons between the proposed CABM scheme and the MV-IGCH scheme when $\varsigma = 0, 50, 100$, $M=3$ and $K=5$ bits.}
\label{fig9}
\end{figure}

\begin{figure}
\centering
\includegraphics[width=8.5cm]{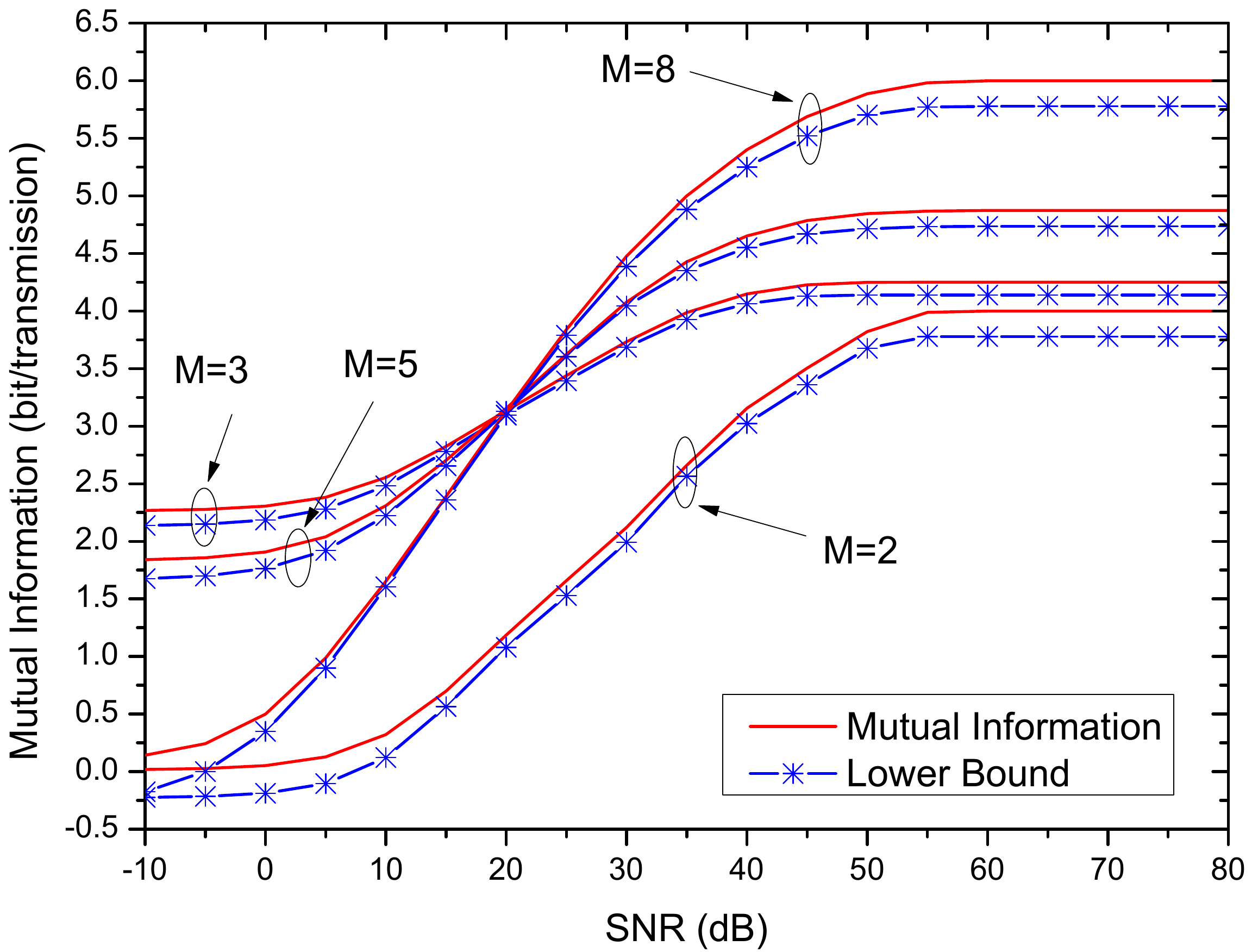}
\caption{Mutual information versus SNR for different $M$ when $q=3$ bits.}
\label{fig10}
\end{figure}

\subsection{Mutual Information Results}
\label{section6_2}
In this subsection, the performance of the mutual information and its lower bound for the SM based VLC system will be shown.
Also, the results for the systems with precoding and without precoding will also be presented.

Fig. \ref{fig10} shows the mutual information versus SNR for different $M$ when $q=3$ bits.
Note that when $M \neq 2^p$ and $M=2^p$, the mutual information results are derived by using (\ref{eq23}) and (\ref{eq27}), respectively;
and the lower bounds of the mutual information are derived by using (\ref{eq34}) and (\ref{eq40}), respectively.
Obviously, the mutual information increases with the increase of SNR.
The gap between the mutual information and its lower bound is very small in the moderate SNR regime.
In the low and high SNR regimes, small constant gaps between the mutual information and its lower bound are shown,
which coincides with that in (\ref{eq37}) and (\ref{eq38}).
Moreover, in low SNR regime, the systems with small $M=3$ and 5 achieve bigger mutual information than that with large $M=2$ and 8.
However, in the high SNR regime, the mutual information increases with the increase of $M$.

\begin{figure}
\centering
\includegraphics[width=8.5cm]{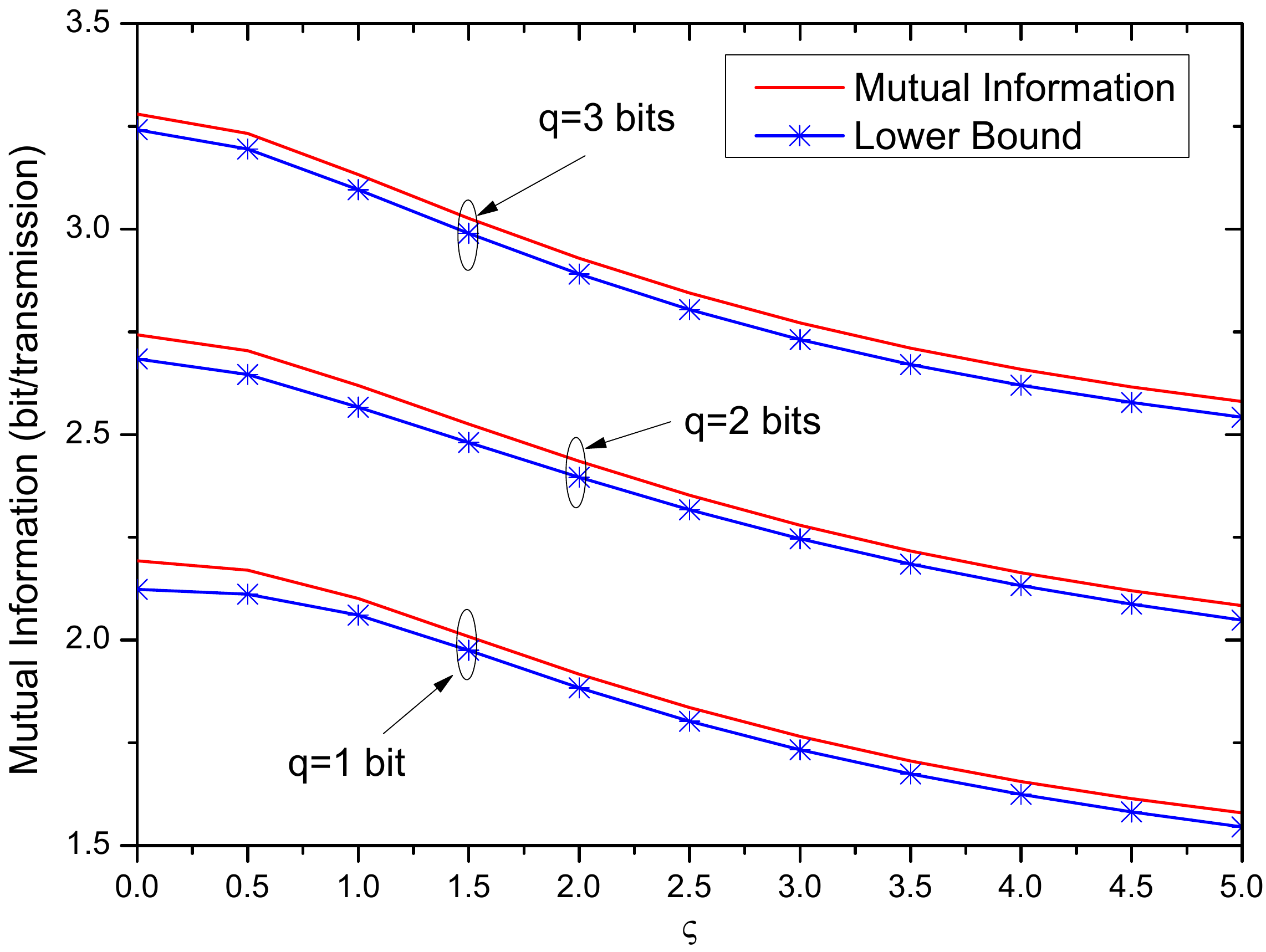}
\caption{Mutual information versus $\varsigma$ for $q=1,2,3$ bits when $M=3$ and SNR=20 dB.}
\label{fig11}
\end{figure}

Fig. \ref{fig11} shows the mutual information versus $\varsigma$ for $q=1,2,3$ bits when $M=3$ and SNR=20 dB.
It can be seen that the mutual information and its lower bound decrease with the increase of $\varsigma$. When $\varsigma=0$, only the input-independent noise is considered, the maximum mutual information is achieved. This indicates that the variance of the input-dependent noise has an important impact on system performance.
Moreover, with the increase of $q$, the data bits transmitted in each time slot become larger, and thus the mutual information also increases.

\begin{figure}
\centering
\includegraphics[width=8.5cm]{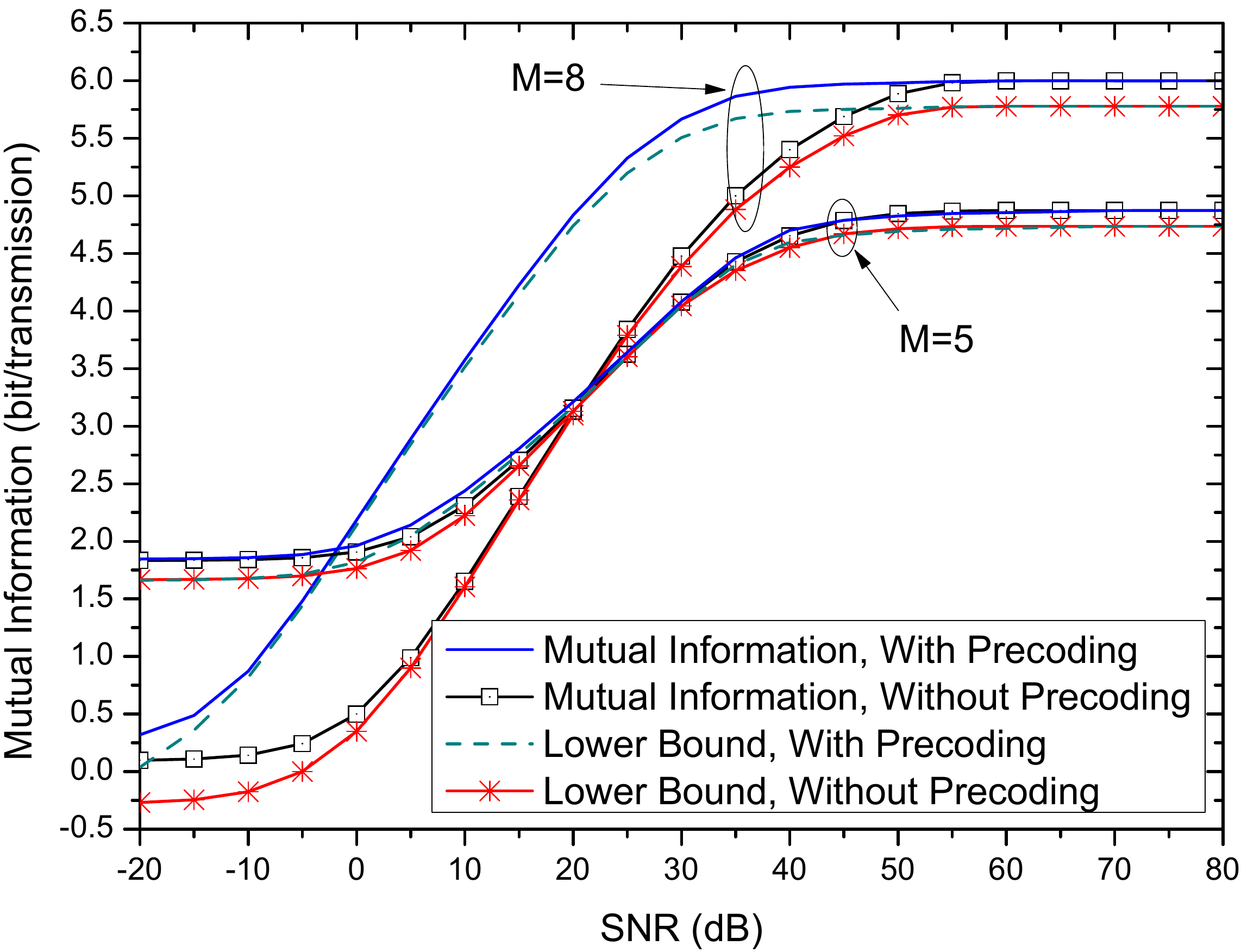}
\caption{Mutual information for schemes with and without precoding when $M=5, 8$ and $q=3$ bits.}
\label{fig12}
\end{figure}

To show the efficiency of the precoding scheme,
Fig. \ref{fig12} shows the mutual information for schemes with and without precoding when $M=5, 8$ and $q=3$ bits.
As can be observed, in the low SNR regime, the system with small $M$ achieves bigger mutual information than that with large $M$.
Moreover, with the increase of SNR, the gap between the systems with small $M$ and with large $M$ becomes smaller and smaller.
In the high SNR regime, the system with large $M$ achieves better performance than that with small $M$.
Furthermore, the system with precoding scheme always achieves better performance than that without using precoding scheme.
This indicates that the proposed precoding scheme is efficient.

\section{Conclusions}
\label{section7}
This paper has investigated the CABM scheme for VLC with an arbitrary number of LEDs.
The main conclusions are summarized as follows:
\begin{enumerate}
  \item A CABM scheme is proposed, which utilizes the CSIT and allows an arbitrary number of LEDs. In traditional bit mapping scheme, the design flexibility is limited by the fact that the number of transmitters must be a power of two. The proposed scheme breaks the limitation on the required number of transmitters. Numerical results show that the proposed CABM scheme outperforms the existing MV-IGCH scheme.
  \item Based on the CABM scheme, the theoretical expression of mutual information for VLC is derived.
        However, it is with an integral expression, which is very hard to evaluate the system performance.
        Alternatively, a lower bound of the mutual information is derived.
        Some insights on the system performance are provided,
        such as the system performance in the high/low SNR regimes, the system performance with $2^p$ LEDs,
        and the gap between mutual information and its lower bound.
  \item A precoding scheme is proposed to maximize the minimum distance between any two constellation points in the received signal space. The optimization problem is transformed by using an approximation and then solved
by using the interior point method. As shown in numerical results, the performance of mutual information is improved by using the precoding scheme.
\end{enumerate}

%\section*{Acknowledgements}
%This work is supported by National Natural Science Foundation of China (61701254 \& 61571115), Natural Science Foundation of Jiangsu Province (BK20170901), and Key International Cooperation Research Project (61720106003), the open research fund of National Mobile Communications Research Laboratory, Southeast University (2017D06),
%the open research fund of Key Lab of Broadband Wireless Communication and Sensor Network Technology (Nanjing University of Posts and Telecommunications), Ministry of Education (JZNY201706), NUPTSF (NY216009),
%the open research fund for Jiangsu Key Laboratory of traffic and transportation security (Huaiyin Institute of Technology) (TTS2017-03),
%the open research fund of Key Laboratory of Intelligent Computing \& Signal Processing (Anhui University),
%the Open Research Subject of Key Laboratory (Research Base) of Signal and Information Processing, Xihua University (szjj2017-047).

%\numberwithin{equation}{section}
\appendices
\section{Proof of \emph{Theorem \ref{the1}}}
\label{appa}
By using (\ref{eq7}), (\ref{eq12}), (\ref{eq13}) and (\ref{eq14}), ${\cal I}(h;y|x)$ in (\ref{eq16}) can be derived as
\begin{eqnarray}
&&\!\!\!\!\!\!\!\!\!\!\!\!\!\!\! {\cal I}\left( {h;y\left| x \right.} \right) = \frac{{{{( {M \!\!-\!\! {2^p}})}^2} \!+\! {{( {{2^{p + 1}} \!\!-\!\! M} )}^2}}}{{{2^{2p}}}}(p \!+\! 1)\nonumber\\
&&\!\!\!\!\!\!\!\!\!\!\!\!\!\!\!+ \frac{{M \!-\! {2^p}}}{{{2^{2p + q}}}}\!\!\! \sum\limits_{m \in \mathbf{\Psi}  \cup \mathbf{\Phi }}\! {\sum\limits_{{x_i} \in \mathbf{A}} \!\!{\mathbb{E}{_z}\!\!\left[\! {{{\log }_2} \frac{{\frac{{\exp \left[ { - \frac{{{{z}^2}}}{{2(1 + {h_m}{x_i}{\varsigma ^2}){\sigma ^2}}}} \right]}}{{\sqrt { 1 + {h_m}{x_i}{\varsigma ^2}} }}}}{{\sum\limits_{{m_2} \in \mathbf{\Psi}  \cup \mathbf{\Phi} } \!\!\!\!{\frac{{\exp \!\left[\! { - \frac{{{{\left( {z+d_m^{m_2}} \right)}^2}}}{{2\left( {1 + {h_{{m_2}}}{x_i}{\varsigma ^2}} \right){\sigma ^2}}}}\! \right]}}{{\sqrt { {1 + {h_{{m_2}}}{x_i}{\varsigma ^2}}} }}} }} }\!\! \right]} } \nonumber \\
 &&\!\!\!\!\!\!\!\!\!\!\!\!\!\!\!+ \frac{{{2^{p + 1}} \!\!-\!\! M}}{{{2^{2p + q}}}}\!\!\!\sum\limits_{m \in \mathbf{\Xi} }\! {\sum\limits_{{x_i} \in \mathbf{B}} \!\! {\mathbb{E}{_z}\!\!\left[\! {{{\log }_2} \frac{{\frac{{\exp \left( { - \frac{{{z^2}}}{{2(1 + {h_m}{x_i}{\varsigma ^2}){\sigma ^2}}}} \right)}}{{\sqrt {1 + {h_m}{x_i}{\varsigma ^2}} }}}}{{\sum\limits_{{m_2} \in \mathbf{\Xi} }\!\! {2\frac{{\exp\! \left[\! { - \frac{{{{\left( z+d_m^{m_2} \right)}^2}}}{{2\left( {1 + {h_{{m_2}}}{x_i}{\varsigma ^2}} \right){\sigma ^2}}}} \!\right]}}{{\sqrt { {1 + {h_{{m_2}}}{x_i}{\varsigma ^2}} } }}} }} }\!\! \right]} }.
\label{eq18}
\end{eqnarray}
where $z = y - {h_m}{x_i}$ and $d_m^{{m_2}} = {h_m}{x_i} - {h_{{m_2}}}{x_i}$.

Furthermore, ${\cal I}\left( {x;y} \right)$ in (\ref{eq16}) can be derived as
\begin{eqnarray}
\!\!\!\!\!\!\!\!\!\!{\cal I}\left( {x;y} \right) \!\!\!\!\!&=&\!\!\!\!\! {\cal H}\left( x \right) - {\cal H}\left( {x\left| y \right.} \right) \nonumber \\
&=&\!\!\!\!\! \frac{{M \!\!-\!\! {2^p}}}{{{2^p}}}{\log _2}\frac{{{2^{p + q - 1}}}}{{M \!\!-\!\! {2^p}}} \!+\! \frac{{{2^{p + 1}} \!\!-\!\! M}}{{{2^p}}}{\log _2}\frac{{{2^{p + q}}}}{{{2^{p + 1}} \!\!-\!\! M}}\nonumber\\
&-&\!\!\!\!\! \frac{{M \!\!-\!\! {2^p}}}{{{2^{2p + q}}}}\!\!\sum\limits_{{x_i} \in {\bf{A}}} \!{\sum\limits_{m \in \mathbf{\Psi}  \cup \mathbf{\Phi}}\!\!\! I_1 } \!-\! \frac{{{2^{p + 1}} \!\!-\!\! M}}{{{2^{2p + q}}}}\!\!\sum\limits_{{x_i} \in {\bf{B}}}\! {\sum\limits_{m \in \mathbf{\Xi} } \!\!I_2 },
\label{eq19}
\end{eqnarray}
where $I_1$ and $I_2$ are given by (\ref{eq19_0}) and (\ref{eq19_1}) as shown at the top of the next page.
\begin{table*}\normalsize
\begin{eqnarray}
I_1\!\!\!\!\!\!&=&\!\!\!\!\!\!\!\!\int_{ - \infty }^{ + \infty }\!\! \frac{{\exp\! \left[\! { - \frac{{{{\left( {y - {h_m}{x_i}} \right)}^2}}}{{2\left( {1 + {h_m}{x_i}{\varsigma ^2}} \right){\sigma ^2}}}} \!\right]}}{{\sqrt {2\pi \left( {1 \!+\! {h_m}{x_i}{\varsigma ^2}} \right){\sigma ^2}} }} {{\log }_2}\frac{{\sum\limits_{{x_{{i_{_2}}}} \in {\bf{A}}}{\sum\limits_{{m_2} \in \mathbf{\Psi}  \cup \mathbf{\Phi} } \!\! {\frac{{\exp \left[\! { - \frac{{{{\left( {y - {h_{{m_2}}}{x_{{i_2}}}} \right)}^2}}}{{2\left( {1 + {h_{{m_2}}}{x_{{i_2}}}{\varsigma ^2}} \right){\sigma ^2}}}}\! \right]}}{{\sqrt {2\pi \left( {1 + {h_{{m_2}}}{x_{{i_2}}}{\varsigma ^2}} \right){\sigma ^2}} }}} }  \!\!+\!\!  \frac{{2^{p + 1}} \!-\! M}{{M - {2^p}}}\!\! \sum\limits_{{x_{{i_{_2}}}} \in {\bf{B}}}\! {\sum\limits_{{m_2} \in \mathbf{\Xi} } \!\!{\frac{{\exp \!\left[ { - \frac{{{{\left( {y - {h_{{m_2}}}{x_{{i_2}}}} \right)}^2}}}{{2\left( {1 + {h_{{m_2}}}{x_{{i_2}}}{\varsigma ^2}} \right){\sigma ^2}}}} \right]}}{{\sqrt {2\pi \left( {1 + {h_{{m_2}}}{x_{{i_2}}}{\varsigma ^2}} \right){\sigma ^2}} }}} } }}{{\sum\limits_{{m_2} \in \mathbf{\Psi}  \cup \mathbf{\Phi} } {\frac{{\exp \left[ { - \frac{{{{\left( {y - {h_{{m_2}}}{x_i}} \right)}^2}}}{{2\left( {1 + {h_{{m_2}}}{x_i}{\varsigma ^2}} \right){\sigma ^2}}}} \right]}}{{\sqrt {2\pi \left( {1 + {h_{{m_2}}}{x_i}{\varsigma ^2}} \right){\sigma ^2}} }}} }}{\rm{d}}y.
\label{eq19_0}
\end{eqnarray}
\hrulefill
\end{table*}

\begin{table*}\normalsize
\begin{eqnarray}
I_2 \!\!=\!\!\!\! \int_{ - \infty }^{ + \infty }\!\! \frac{{\exp\! \left[\! { - \frac{{{{\left( {y - {h_m}{x_i}} \right)}^2}}}{{2\left( {1 + {h_m}{x_i}{\varsigma ^2}} \right){\sigma ^2}}}}\! \right]}}{{\sqrt {2\pi \left( {1 \!+\! {h_m}{x_i}{\varsigma ^2}} \right){\sigma ^2}} }} \log_2\!\frac{{ \frac{M - {2^p}}{{2^{p + 1}} \!-\! M} \!\!\sum\limits_{{x_{{i_{_2}}}} \in {\bf{A}}}\! {\sum\limits_{{m_2} \in \mathbf{\Psi}  \cup \mathbf{\Phi} }\!\!\! {\frac{{\exp \!\left[\! { - \frac{{{{\left( {y - {h_{{m_2}}}{x_{{i_2}}}} \right)}^2}}}{{2\left( {1 + {h_{{m_2}}}{x_{{i_2}}}{\varsigma ^2}} \right){\sigma ^2}}}}\! \right]}}{{\sqrt {2\pi \left( {1 + {h_{{m_2}}}{x_{{i_2}}}{\varsigma ^2}} \right){\sigma ^2}} }}} }  \!\!+\!\! \sum\limits_{{x_{{i_{_2}}}} \in {\bf{B}}} {\sum\limits_{{m_2} \in \mathbf{\Xi} } {\frac{{\exp \left[ { - \frac{{{{\left( {y - {h_{{m_2}}}{x_{{i_2}}}} \right)}^2}}}{{2\left( {1 + {h_{{m_2}}}{x_{{i_2}}}{\varsigma ^2}} \right){\sigma ^2}}}} \right]}}{{\sqrt {2\pi \left( {1 + {h_{{m_2}}}{x_{{i_2}}}{\varsigma ^2}} \right){\sigma ^2}} }}} } }}{{\sum\limits_{{m_2} \in \mathbf{\Xi} } {\frac{{\exp \left[ { - \frac{{{{\left( {y - {h_{{m_2}}}{x_i}} \right)}^2}}}{{2\left( {1 + {h_{{m_2}}}{x_i}{\varsigma ^2}} \right){\sigma ^2}}}} \right]}}{{\sqrt {2\pi \left( {1 + {h_{{m_2}}}{x_i}{\varsigma ^2}} \right){\sigma ^2}} }}} }}{\rm{d}}y.
\label{eq19_1}
\end{eqnarray}
\hrulefill
\end{table*}

Let $d_{m,i}^{{m_2},{i_2}} = {h_m}{x_i} - {h_{{m_2}}}{x_{{i_2}}}$ and $d_m^{{m_2}} = {h_m}{x_i} - {h_{{m_2}}}{x_i}$, $I_1$ and $I_2$ can be further written, respectively, as
\begin{eqnarray}
&&\!\!\!\!\!\!\!\!\!\!\!\!\!\!\!\!\! {I_1} \!\!=\!\! \mathbb{E}{_z}\!\!\left\{\!\! {{{\log }_2}\!\!\left\{\!\! {\left[\! {\left( {M \!\!-\!\! {2^p}} \right)\!\!\!\sum\limits_{{x_{{i_2}}} \in {\bf{A}}}\! {\sum\limits_{{m_2} \in \bf{\Psi}  \cup \bf{\Phi} }\!\!\!\! {\exp \!\!\left[\!\! {  \frac{{{{-\left( {z + d_{m,i}^{{m_2},{i_2}}} \right)}^2}}}{{2\!\left( {1 \!\!+\!\! {h_{{m_2}}}{x_{{i_2}}}{\varsigma ^2}} \right)\!{\sigma ^2}}}}\!\! \right]} } } \right.} \right.} \right. \nonumber \\
&&\!\!\!\!\!\!\!\!\!\!\!\!\!\!\!\!\! {{\left. { +\!\! \left( {{2^{p + 1}} \!\!-\!\! M} \right)\!\!\!\sum\limits_{{x_{{i_2}}} \in {\bf{B}}} {\sum\limits_{{m_2} \in \bf{\Xi} }\! {\exp \!\!\left[\! {  \frac{{{{-\left( {z + d_{m,i}^{{m_2},{i_2}}} \right)}^2}}}{{2\!\left( {1 \!+\! {h_{{m_2}}}{x_{{i_2}}}{\varsigma ^2}} \right)\!{\sigma ^2}}}} \!\right]}\! } } \right]} \mathord{\left/
 {\vphantom {{\left. {\;\;\;\;\; + \left( {{2^{p + 1}} - M} \right)\sum\limits_{{x_{{i_2}}} \in {\bf{B}}} {\sum\limits_{{m_2} \in \Xi } {\exp \left( { - \frac{{{{\left( {z + d_{m,i}^{{m_2},{i_2}}} \right)}^2}}}{{2\left( {1 + {h_{{m_2}}}{x_{{i_2}}}{\varsigma ^2}} \right){\sigma ^2}}}} \right)} } } \right]} {}}} \right.
 \kern-\nulldelimiterspace} {}}\nonumber \\
&&\!\!\!\!\!\!\!\!\!\!\!\!\!\!\!\!\! \left. {\left. {\left[\!\! {(\! {M \!\!-\!\! {2^p}}\!)\!\!\frac{{\sqrt {1 \!\!+\!\! {h_{{m_2}}}{x_{{i_2}}}{\varsigma ^2}} }}{{\sqrt {1 \!\!+\!\! {h_{{m_2}}}{x_i}{\varsigma ^2}} }}\!\!\!\sum\limits_{{m_2} \in \bf{\Psi}  \cup \bf{\Phi} }\!\!\!\! {\exp\!\! \left[\! { \frac{{{{-\left( {z + d_m^{{m_2}}} \right)}^2}}}{{2\left( {1 + {h_{{m_2}}}{x_i}{\varsigma ^2}} \right){\sigma ^2}}}} \!\!\right]} }\!\! \right]} \!\!\right\}}\!\! \right\}\!\!,
\label{eq20}
\end{eqnarray}
and
\begin{eqnarray}
&&\!\!\!\!\!\!\!\!\!\!\!\!\!\!\!{I_2} \!\!=\!\! { \mathbb{E}_z}\!\!\left\{\!\! {{{\log }_2}\!\!\left\{\!\! {\left[\!\! {\left(\! {M \!-\! {2^p}}\! \right)\!\!\!\sum\limits_{{x_{{i_2}}} \in {\bf{A}}}\! {\sum\limits_{{m_2} \in {\bf{\Psi }} \cup {\bf{\Phi }}}\!\!\! {\exp\!\! \left[\! { \frac{{{{-\left( {z + d_{m,i}^{{m_2},{i_2}}} \right)}^2}}}{{2\!\left( {1 \!+\! {h_{{m_2}}}{x_{{i_2}}}{\varsigma ^2}} \right)\!{\sigma ^2}}}}\!\! \right]} } } \right.} \right.} \right. \nonumber \\
&&\!\!\!\!\!\!\!\!\!\!\!\!\!\!\!{{\left. { +\!\! \left( {{2^{p + 1}} \!\!-\!\! M} \right)\!\!\sum\limits_{{x_{{i_2}}} \in {\bf{B}}} \!{\sum\limits_{{m_2} \in {\bf{\Xi }}} \!\!{\exp\!\! \left[\! {  \frac{{{{-\left( {z + d_{m,i}^{{m_2},{i_2}}} \right)}^2}}}{{2\!\left( {1 \!+\! {h_{{m_2}}}{x_{{i_2}}}{\varsigma ^2}} \right)\!{\sigma ^2}}}}\! \right]} } } \right]} \mathord{\left/
 {\vphantom {{\left. { + \left( {{2^{p + 1}} - M} \right)\sum\limits_{{x_{{i_2}}} \in {\bf{B}}} {\sum\limits_{{m_2} \in {\bf{\Xi }}} {\exp \left( { - \frac{{{{\left( {z + d_{m,i}^{{m_2},{i_2}}} \right)}^2}}}{{2\left( {1 + {h_{{m_2}}}{x_{{i_2}}}{\varsigma ^2}} \right){\sigma ^2}}}} \right)} } } \right]} {}}} \right.
 \kern-\nulldelimiterspace} {}}\nonumber \\
&&\!\!\!\!\!\!\!\!\!\!\!\!\!\!\!\left. {\left. {\left[\!\! {\left( {{2^{p + 1}} \!\!-\!\! M} \right)\!\!\frac{{\sqrt {\!1 \!\!+\!\! {h_{{m_2}}}{x_{{i_2}}}{\varsigma ^2}} }}{{\sqrt {\!1 \!\!+\!\! {h_{{m_2}}}{x_i}{\varsigma ^2}} }}\!\!\!\sum\limits_{{m_2} \in {\bf{\Xi }}}\!\! {\exp\!\! \left[\! {  \frac{{{{-\left( {z + d_m^{{m_2}}} \right)}^2}}}{{2\!\left( {1 \!+\! {h_{{m_2}}}{x_i}{\varsigma ^2}} \right)\!{\sigma ^2}}}} \!\! \right]} }\!\! \right]} \!\!\right\}}\!\! \right\}\!\!.
\label{eq21}
\end{eqnarray}
Substitute (\ref{eq20}) and (\ref{eq21}) into (\ref{eq19}), ${\cal I}\left( {x;y} \right)$ can be finally written as (\ref{eq22}) as shown at the top of the next page.
\begin{table*}\normalsize
\begin{eqnarray}
&&\!\!\!\!\!\!\!\!\!\!{\cal I}\left( {x;y} \right) = \frac{{M - {2^p}}}{{{2^p}}}{\log _2}\frac{{{2^{p + q - 1}}}}{{M - {2^p}}} + \frac{{{2^{p + 1}} - M}}{{{2^p}}}{\log _2}\frac{{{2^{p + q}}}}{{{2^{p + 1}} - M}}\nonumber\\
&&\!\!\!\!\!\!\!\!\!\!- \frac{{M \!-\! {2^p}}}{{{2^{2p + q}}}}\!\!\sum\limits_{{x_i} \in {\bf{A}}}\! \sum\limits_{m \in\mathbf{ \Psi}  \cup \mathbf{\Phi }}\!\!\!\mathbb{E}{_z}\!\!\left[\! {{{\log }_2}\frac{{\sum\limits_{{x_{{i_2}}} \in {\bf{A}}}\! {\sum\limits_{{m_2} \in \mathbf{\Psi } \cup \mathbf{\Phi} }\!\!\!\! \frac{\exp \!\left[\!\! { - \frac{{{{\left( {z + d_{m,i}^{{m_2},{i_2}}} \right)}^2}}}{{2\left( {1 + {h_{{m_2}}}{x_{{i_2}}}{\varsigma ^2}} \right){\sigma ^2}}}} \!\!\right]}{{\sqrt {1 + {h_{{m_2}}}{x_{{i_2}}}{\varsigma ^2}} }} } \!+\! \frac{{{2^{p + 1}} \!-\!M}}{ {M - {2^p}}}\!\!\!\! \sum\limits_{{x_{{i_2}}} \in {\bf{B}}}\!{\sum\limits_{{m_2} \in \mathbf{\Xi} }\!\! \frac{\exp \!\left[\! { - \frac{{{{\left( {z + d_{m,i}^{{m_2},{i_2}}} \right)}^2}}}{{2\left( {1 + {h_{{m_2}}}{x_{{i_2}}}{\varsigma ^2}} \right){\sigma ^2}}}} \!\!\right]}{{\sqrt {1 + {h_{{m_2}}}{x_{{i_2}}}{\varsigma ^2}} }} } }}{{\sum\limits_{{m_2} \in \mathbf{\Psi}  \cup \mathbf{\Phi} } \frac{\exp \left[ { - \frac{{{{\left( {z + d_m^{{m_2}}} \right)}^2}}}{{2\left( {1 + {h_{{m_2}}}{x_i}{\varsigma ^2}} \right){\sigma ^2}}}} \right]}{{\sqrt {1 + {h_{{m_2}}}{x_i}{\varsigma ^2}} }} }}}\! \right] \nonumber \\
&&\!\!\!\!\!\!\!\!\!\!- \frac{{{2^{p \!+\! 1}} \!-\! M}}{{{2^{2p + q}}}}\!\!\sum\limits_{{x_i} \in {\bf{B}}} \!\sum\limits_{m \in \mathbf{\Xi} } \!\!\mathbb{E}{_z}\!\!\left[\!  {{{\log }_2}\frac{{ \frac{M - {2^p}}{{{2^{p \!+\! 1}} \!-\! M}} \!\!\! \sum\limits_{{x_{{i_2}}} \in {\bf{A}}} \!{\sum\limits_{{m_2} \in \mathbf{\Psi } \cup \mathbf{\Phi} } \!\!\!\frac{\exp\! \left[\!\! { - \frac{{{{\left( {z + d_{m,i}^{{m_2},{i_2}}} \right)}^2}}}{{2\left(\! {1 \!+\! {h_{{m_2}}}{x_{{i_2}}}{\varsigma ^2}} \!\right){\sigma ^2}}}}\!\! \right]}{{\sqrt {1 + {h_{{m_2}}}{x_{{i_2}}}{\varsigma ^2}} }} } \!+\!\!\! \sum\limits_{{x_{{i_2}}} \in {\bf{B}}} \!{\sum\limits_{{m_2} \in \mathbf{\Xi} }\!\!\!\frac{\exp\! \left[\!\! { - \frac{{{{\left( {z + d_{m,i}^{{m_2},{i_2}}} \right)}^2}}}{{2\left(\! {1 + {h_{{m_2}}}{x_{{i_2}}}{\varsigma ^2}} \! \right){\sigma ^2}}}} \!\!\right]}{{\sqrt {1 + {h_{{m_2}}}{x_{{i_2}}}{\varsigma ^2}} }} } }}{{\sum\limits_{{m_2} \in \mathbf{\Xi} } \frac{\exp \left( { - \frac{{{{\left( {z + d_m^{{m_2}}} \right)}^2}}}{{2\left( {1 + {h_{{m_2}}}{x_i}{\varsigma ^2}} \right){\sigma ^2}}}} \right)}{{\sqrt {1 + {h_{{m_2}}}{x_i}{\varsigma ^2}} }} }}}\!\! \right].
 \label{eq22}
\end{eqnarray}
\hrulefill
\end{table*}
Furthermore, substitute (\ref{eq18}) and (\ref{eq22}) into (\ref{eq16}), (\ref{eq23}) can be derived.

\section{Proof of \emph{Theorem \ref{the2}}}
\label{appb}
To facilitate the analysis, (\ref{eq23}) can be rewritten as
\begin{eqnarray}
&&\!\!\!\!\!\! {\cal I}\left( {x,h;y} \right) \!=\! \frac{{{{\left( {M \!-\! {2^p}} \right)}^2} \!+\! {{\left( {{2^{p + 1}} \!-\! M} \right)}^2}}}{{{2^{2p}}}}\left( {p \!+\! 1} \right) \nonumber\\
&&\!\!\!\!\!\!+ \frac{{M \!-\! {2^p}}}{{{2^p}}}{\log _2}\frac{{{2^{p + q - 1}}}}{{M \!-\! {2^p}}} \!+\! \frac{{{2^{p + 1}} \!-\! M}}{{{2^p}}}{\log _2}\frac{{{2^{p + q}}}}{{{2^{p + 1}} \!-\! M}} \nonumber \\
&&\!\!\!\!\!\! - \frac{{M \!-\! {2^p}}}{{{2^{2p + q}}}}\!\!\!\sum\limits_{m \in \mathbf{\Psi}  \cup \mathbf{\Phi }} {\sum\limits_{{x_i} \in {\bf{A}}} {\underbrace {{\mathbb{E}_z}\!\left[ {{{\log }_2}\!\!\left[ {\exp\!\! \left(\! {  \frac{{{z^2}}}{{2(1 \!+\! {h_m}{x_i}{\varsigma ^2}){\sigma ^2}}}}\! \right)} \right]} \right]}_{{I_3}}} }  \nonumber\\
&&\!\!\!\!\!\!- \frac{{{2^{p + 1}} \!-\! M}}{{{2^{2p + q}}}}\!\!\sum\limits_{m \in \mathbf{\Xi} } {\sum\limits_{{x_i} \in {\bf{B}}} \!\!{{\mathbb{E}_z}\!\!\left[\! {{{\log }_2}\!\!\left[ {\exp\!\! \left(\! {  \frac{{{z^2}}}{{2(1 \!+\! {h_m}{x_i}{\varsigma ^2}){\sigma ^2}}}}\! \right)}\! \right]}\! \right]} } \nonumber\\
&&\!\!\!\!\!\! - \frac{{M \!-\! {2^p}}}{{{2^{2p + q}}}}\!\!\sum\limits_{m \in \mathbf{\Psi}  \cup \mathbf{\Phi} }\! \sum\limits_{{x_i} \in {\bf{A}}} I_4 - \frac{{{2^{p + 1}} - M}}{{{2^{2p + q}}}}\sum\limits_{m \in \mathbf{\Xi} } \sum\limits_{{x_i} \in {\bf{B}}} I_5,
 \label{eq30}
\end{eqnarray}
where $I_4$ and $I_5$ are given, respectively, as
\begin{eqnarray}
I_4\!\!\!\!\!\!&=&\!\!\!\!\!\!{\mathbb{E}_z}\!\!\left\{\!\! {{{\log }_2}\!\!\!\left[\!\! {\sqrt {\!1 \!\!+\!\! {h_m}{x_i}{\varsigma ^2}}\!\! \left(\!\! {\sum\limits_{{x_{{i_2}}} \in {\bf{A}}} \!{\sum\limits_{{m_2} \in \mathbf{\Psi}  \cup \mathbf{\Phi} } \!\!\!\!{\frac{{\exp\!\! \left[\! {  \frac{{{{-\left( {z + d_{m,i}^{{m_2},{i_2}}} \right)}^2}}}{{2\left( {1 + {h_{{m_2}}}{x_{{i_2}}}{\varsigma ^2}} \right){\sigma ^2}}}} \!\right]}}{{\sqrt {1 \!+\! {h_{{m_2}}}{x_{{i_2}}}{\varsigma ^2}} }}} } } \right.} \right.} \right.\nonumber\\
&+&\!\!\!\!\!\!\!\left. {\left. {\left. { \frac{{{2^{p + 1}} \!-\! M}}{{M \!-\! {2^p}}}\!\!\!\sum\limits_{{x_{{i_2}}} \in {\bf{B}}} {\sum\limits_{{m_2} \in \mathbf{\Xi} } {\frac{{\exp \!\!\left[\! {  \frac{{{{-\left( {z + d_{m,i}^{{m_2},{i_2}}} \right)}^2}}}{{2\left( {1 + {h_{{m_2}}}{x_{{i_2}}}{\varsigma ^2}} \right){\sigma ^2}}}} \!\right]}}{{\sqrt {1 + {h_{{m_2}}}{x_{{i_2}}}{\varsigma ^2}} }}} } } \right)} \right]} \right\},
\label{eq30_0}
\end{eqnarray}
and
\begin{eqnarray}
{I_5}\!\!\!\!\! &=&\!\!\!\!\! {\mathbb{E}_z}\left\{ {{{\log }_2}\left[ {2\sqrt {1 + {h_m}{x_i}{\varsigma ^2}} \left( {\frac{{M - {2^p}}}{{{2^{p + 1}} - M}} } \right.} \right.} \right.\nonumber\\
&\times&\!\!\!\!\! \sum\limits_{{x_{{i_2}}} \in {\bf{A}}} {\sum\limits_{{m_2} \in \mathbf{\Psi}  \cup \mathbf{\Phi} } {\frac{{\exp \left( { - \frac{{{{\left( {z + d_{m,i}^{{m_2},{i_2}}} \right)}^2}}}{{2\left( {1 + {h_{{m_2}}}{x_{{i_2}}}{\varsigma ^2}} \right){\sigma ^2}}}} \right)}}{{\sqrt {1 + {h_{{m_2}}}{x_{{i_2}}}{\varsigma ^2}} }}} }\nonumber\\
&+& \!\!\!\!\!\left. {\left. {\left. { \sum\limits_{{x_{{i_2}}} \in {\bf{B}}} {\sum\limits_{{m_2} \in \mathbf{\Xi} } {\frac{{\exp \left( { - \frac{{{{\left( {z + d_{m,i}^{{m_2},{i_2}}} \right)}^2}}}{{2\left( {1 + {h_{{m_2}}}{x_{{i_2}}}{\varsigma ^2}} \right){\sigma ^2}}}} \right)}}{{\sqrt {1 + {h_{{m_2}}}{x_{{i_2}}}{\varsigma ^2}} }}} } } \right)} \right]} \right\}.
\label{eq30_1}
\end{eqnarray}

For $I_3$ in (\ref{eq30}), we have
\begin{eqnarray}
{I_3} =  \frac{1}{2}{\log _2}e.
 \label{eq31}
\end{eqnarray}
For $I_4$ in (\ref{eq30_0}), an upper bound can be derived as
\begin{eqnarray}
&&\!\!\!\!\!\!\!\!\!\!\!\!{I_4}\!\!\le\!\! {\log _2}\!\!\!\left[\!\! {\sum\limits_{{x_{{i_2}}} \!\in\! \mathbf{A}}\! {\sum\limits_{{m_2} \!\in\! \mathbf{\Psi}  \cup \mathbf{\Phi} }\!\!\! {\frac{{\sqrt {1 \!\!+\!\! {h_m}{x_i}{\varsigma ^2}} }}{{\sqrt {\!2\!\left( {1 \!\!+\!\! {h_{{m_2}}}{x_{{i_2}}}{\varsigma ^2}} \right)} }}\!\!\exp\!\!\!\left[\!\! {  \frac{{{{-\left(\! {d_{m,i}^{{m_2},{i_2}}}\! \right)}^2}}}{{4\!\left(\! {1 \!\!+\!\! {h_{{m_2}}}\!{x_{{i_2}}}\!{\varsigma ^2}} \! \right)\!{\sigma ^2}}}}\!\! \right]} } } \right.\nonumber\\
&&\!\!\!\!\!\!\!\!\!\!\!\!\left. { +\! \frac{{{2^{p \!+\! 1}} \!\!-\!\! M}}{{M \!\!-\!\! {2^p}}}\!\!\!\!\sum\limits_{{x_{{i_2}}} \!\in\! \mathbf{B}} \! {\sum\limits_{{m_2} \!\in\! \mathbf{\Xi} }\!\! {\frac{{\sqrt {1 \!\!+\!\! {h_m}{x_i}{\varsigma ^2}} }}{{\sqrt {\!2\!\left(\! {1 \!\!+\!\! {h_{{m_2}}}\!{x_{{i_2}}}\!{\varsigma ^2}}\! \right)} }}} \!\exp\!\!\!\left[\!\! {  \frac{{{{-\left( {d_{m,i}^{{m_2},{i_2}}} \right)}^2}}}{{4\!\left(\! {1 \!\!+\!\! {h_{{m_2}}}\!{x_{{i_2}}}\!{\varsigma ^2}}\! \right)\!{\sigma ^2}}}}\!\! \right]} }\!\! \right]\!\!\!.
\label{eq32}
\end{eqnarray}
Similarly, for $I_5$ in (\ref{eq30_1}), an upper bound is given by
\begin{eqnarray}
{I_5} \!\!\!\!\! &\le&\!\!\!\!\! {\log _2}\left[ {\frac{{M - {2^p}}}{{{2^{p + 1}} - M}}\sum\limits_{{x_{{i_2}}} \in {\bf{A}}} {\sum\limits_{{m_2} \in \mathbf{\Psi}  \cup \mathbf{\Phi} } \!\!\!{\frac{{\sqrt {2\left( {1 + {h_m}{x_i}{\varsigma ^2}} \right)} }}{{\sqrt {1 + {h_{{m_2}}}{x_{{i_2}}}{\varsigma ^2}} }}} } } \right.\nonumber\\
&\times&\!\!\!\!\! \exp\left[ { - \frac{{{{\left( {d_{m,i}^{{m_2},{i_2}}} \right)}^2}}}{{4\left( {1 + {h_{{m_2}}}{x_{{i_2}}}{\varsigma ^2}} \right){\sigma ^2}}}} \right]\nonumber\\
&+&\!\!\!\!\!\!\!\!\!\left. {  \sum\limits_{{x_{{i_2}}} \!\in\! {\bf{B}}} {\sum\limits_{{m_2} \!\in\! \mathbf{\Xi }}\!\!\! {\frac{{\sqrt {2\!\left( {1 \!\!+\!\! {h_m}{x_i}{\varsigma ^2}} \right)} }}{{\sqrt {1 \!+\! {h_{{m_2}}}{x_{{i_2}}}{\varsigma ^2}} }}} \!\exp\!\!\!\left[\!\! {  \frac{{{{-\left( {d_{m,i}^{{m_2},{i_2}}} \right)}^2}}}{{4\!\left( {1 \!\!+\!\! {h_{{m_2}}}{x_{{i_2}}}{\varsigma ^2}} \right)\!{\sigma ^2}}}}\!\! \right]} }\! \right].
 \label{eq33}
\end{eqnarray}
Substituting (\ref{eq31})-(\ref{eq33}) into (\ref{eq30}), (\ref{eq34}) is obtained.

\ifCLASSOPTIONcaptionsoff
  \newpage
\fi

\end{document}